\documentclass{amsart}
\pdfoutput=1

\usepackage{amssymb}
\usepackage{amsthm}
\usepackage{amsfonts}
\usepackage{amsmath}
\usepackage{pmboxdraw}
\usepackage{verbatim} % for comment
\usepackage{graphicx}
\usepackage{color}
\usepackage[colorlinks=true, citecolor=blue, filecolor=black, linkcolor=black, urlcolor=black]{hyperref}
\usepackage{cite}
\usepackage[normalem]{ulem}
\usepackage{subcaption}
\usepackage{bbm}
\usepackage{bm}
\usepackage{mathtools}
\usepackage{todonotes}
\usepackage{kantlipsum}
\allowdisplaybreaks
\usepackage{dsfont}
\usepackage[percent]{overpic}

\usepackage{tikz}
\usetikzlibrary{patterns, positioning, arrows.meta,calc}
\usetikzlibrary{patterns,decorations.pathreplacing,decorations}

\newcommand{\RN}[1]{%
	\textup{\uppercase\expandafter{\romannumeral#1}}%
}

\def\C{\mathbb{C}}

\def\P{\mathbf{P}}
\def\R{\mathbb{R}}

\newcommand{\Ai}{\operatorname{Ai}}

\newcommand{\re}{\operatorname{Re}}

\theoremstyle{plain}
\newtheorem{thm}{Theorem}[section]

\newtheorem{cor}[thm]{Corollary}
\newtheorem{lem}[thm]{Lemma}

\newtheorem{prop}[thm]{Proposition}

\theoremstyle{remark}

\newtheorem{rem}{Remark}

\numberwithin{equation}{section}

\begin{document}

\title[Dissipative Spectral Form Factor of the Complex Elliptic Ginibre Ensemble]{Dissipative Spectral Form Factor of the Complex Elliptic Ginibre Ensemble across Various Non-Hermiticity Regimes}
%%%%%%%%%%%%%%%%%%%%%%%%%%%%% author %%%%%%%%%%%%%%%%%%%%%%%%%%%%
\author{Gernot Akemann}
 \address{Faculty of Physics, Bielefeld University, P.O. Box 100131, 33501 Bielefeld, Germany}
 \email{akemann@physik.uni-bielefeld.de}

\author{Sung-Soo Byun}
 \address{Department of Mathematical Sciences and Research Institute of Mathematics, Seoul National University, Seoul 151-747, Republic of Korea}
 \email{sungsoobyun@snu.ac.kr}

\author{Seungjoon Oh}
 \address{Department of Mathematical Sciences, Seoul National University, Seoul 151-747, Republic of Korea}
 \email{seungjoonoh@snu.ac.kr}

%%%%%%%%%%%%%%%%%%%%%%%%%%%%% author %%%%%%%%%%%%%%%%%%%%%%%%%%%%

\begin{abstract} 
We study the dissipative spectral form factor (DSFF) at complex time $T e^{i\theta}$ for the complex elliptic Ginibre ensemble with non-Hermiticity parameter $\tau \in [0,1)$. As the matrix dimension $N \to \infty$, we consider the natural scalings in both the time variable and the non-Hermiticity parameter, namely $T = O(N^\gamma)$ and $1 - \tau = O(N^{-\alpha})$.  For all regimes $\gamma \ge 0$ and $\alpha \ge 0$, we derive the precise asymptotic behaviour of both the disconnected and connected components of the DSFF. In particular, we explicitly characterise the dip--ramp--plateau structure, including the dip time and the Heisenberg time. In addition, we identify the mesoscopic regime $\alpha \in (0,1)$, which interpolates between the behaviour of the DSFF of non-Hermitian random matrices and the spectral form factor (SFF) of Hermitian ensembles. We further provide an explicit description of the phase diagram, in which the ramp exhibits quadratic, linear, or intermediate behaviour depending on the scaling parameters.  
\end{abstract}

\maketitle

%\tableofcontents
 
\section{Introduction}

The spectral statistics of quantum systems have long been shaped by two fundamental conjectures. 
The Berry--Tabor conjecture asserts that generic quantum integrable systems exhibit locally Poissonian level statistics, reflecting the absence of significant correlations between neighbouring energy levels. 
By contrast, the Bohigas--Giannoni--Schmit conjecture predicts that quantum systems whose classical counterparts are chaotic display local level statistics governed by Hermitian random matrix theory (RMT), in accordance with the relevant symmetry class.  
This paradigm of quantum chaos has been extended to open and dissipative quantum systems, where the spectrum becomes complex. In this setting, generic integrable systems are expected to exhibit two-dimensional (2D) Poisson statistics, whereas 
open, dissipative 
quantum chaotic systems display correlated spectra described by non-Hermitian RMT~\cite{GHS88}. 
In this work, we focus on the latter regime.

There are different tools available from RMT to quantify such statistics, on different scales. On a local or microscopic scale, the nearest-neighbour spacing distribution (in radial distance in 2D) is amongst the most popular, due to the Wigner surmise for RMT in 1D. Much more recently, the consecutive spacing ratio has been introduced \cite{OH2007,ABGR2013} in 1D and in \cite{SRP2020} in 2D, which is particularly convenient in comparison to data. Here, the dependence on the local scale, given by the inverse mean spectral density, drops out without so-called unfolding \cite{Ha10}. Nevertheless, little rigorous knowledge is available for this quantity so far, see however \cite{Nis2024,Bui2026} in 1D, even if a surmise is also effective in the Hermitian case in 1D \cite{ABGR2013}. 
This is despite the fact that, in principle, all local quantities should be accessible from knowledge of the local $k$-point functions. At least in the classical Hermitian ensembles of RMT the latter follow from standard asymptotics of classical orthogonal polynomials and their corresponding kernel, being given by determinantal or Pfaffian point processes \cite{Fo10}.

An entirely different tool that probes spectral correlations on all scales—microscopic, mesoscopic, and macroscopic—is the spectral form factor (SFF)~\cite{LLJP86}, defined as the Fourier transform of the two-point function. 
While spectra are typically not universal on macroscopic scales, quantum chaotic systems nonetheless exhibit the characteristic \textit{dip--ramp--plateau} profile predicted by Hermitian RMT; see~\cite{Ha10} and references therein. 
The literature on the SFF is extensive, and we refrain from attempting a comprehensive survey. 
Perhaps surprisingly, the mathematically rigorous literature on the SFF remains limited; see, for instance, \cite{Fo21,Fo21a,FKLZ24,CES23} and 
%the 
references therein, with early contributions going back to \cite{BH97}. 
This is largely due to the difficulty of obtaining rigorous asymptotics, as multiple scales arise and must be matched. 
To complete the picture, one may also consider other tools probing multiple scales, such as the number variance and related quantities~\cite[Chapter~16]{Mehta2004}, as well as linear statistics more generally.

%An entirely different tool that probes spectral correlations on all scales, micro-, meso- and macroscopic,  is the spectral form factor (SFF) \cite{LLJP86}, being given by the Fourier transform of the two-point function.  Even if on a macroscopic scale spectra are typically not universal, quantum chaotic systems typically show the characteristic dip–ramp–plateau profile following from Hermitian RMT, see \cite{Ha10} and references therein. There is a huge literature on the SFF, and it is very difficult to give credit and to list all applications. Perhaps surprisingly, there is little mathematically rigorous literature on the SFF, see \cite{Fo21,Fo21a} and references therein, going back to \cite{BH97}. The reason being that it is much more difficult to provide rigorous asymptotics in this case, as different scales occur, that have to be matched. To complete the picture, also other tools for multiple scales are available, such as the number variance and others \cite[Chap. 16]{Mehta2004}, or linear statistics in general. 

In this work, we focus on the recently introduced dissipative spectral form factor (DSFF)~\cite{LPC21}, which has been found to be applicable to a large variety of open, dissipative quantum chaotic systems~\cite{LYPC24}. 
The DSFF is defined as a two-variable Fourier transform of the two-point function of complex eigenvalues in the appropriate non-Hermitian symmetry class. 
It also displays a dip--ramp--plateau behaviour; however, the ramp is quadratic instead.  
The universality of the results in~\cite{LPC21} has been established for real and complex non-Hermitian random matrices with i.i.d.\ entries~\cite{CG24}. 
Note that different proposals have also been made, e.g.\ the dissipative form factor~\cite{Can2019}, a one-parameter Fourier transform of the two-point function of the Liouvillian (rather than the Hamiltonian), and a fidelity-type deformation~\cite{Aurelia23}.

Given that the study of open dissipative systems in comparison with non-Hermitian RMT has become very popular, see~\cite[Section~6.2]{BF25} for a recent review of quantum chaos from that perspective, our goal is to derive mathematically rigorous asymptotics of the DSFF in 2D, and its interpolation to the SFF in 1D proposed in~\cite{SKSK24}. 
For that purpose, we consider the simplest ensemble without further symmetries, introduced by Ginibre~\cite{Ginibre}, with complex Gaussian matrix entries (GinUE). 
Its one-parameter-dependent extension by Girko~\cite{Gir86} and Sommers et al.~\cite{SCSS88}, namely the elliptic (or Ginibre--Girko) ensemble, denoted by eGinUE, furthermore allows one to interpolate between the strongly non-Hermitian Ginibre case and the Hermitian limit given by the Gaussian Unitary Ensemble (GUE), thereby recovering the SFF in 1D. 
Of particular importance here is the weak non-Hermiticity regime introduced by Fyodorov, Khoruzhenko, and Sommers~\cite{FKS97a,FKS97}.  
Previous heuristic results on this interpolation between DSFF and SFF exist in~\cite{SKSK24}. 
Apart from making these results mathematically rigorous, we also uncover several different scaling regimes for the dip, ramp, and plateau regions, in analogy with recently discovered weak non-Hermiticity regimes for the number variance in the eGinUE~\cite{ADM24}.

%\\

%The remainder of this work is organised as follows. In the remained of this section we will first introduce the eGinUE and define the DSFF. We will then summarise the different scaling regimes that we find (see Fig. 1) and propose several theorems for the according limits of the DSFF for the dip-ramp regime, including its connected and disconnected part. Section 2 is presenting the finite-$N$ calculation, based on the determinantal point process of the eGInUE in terms of planary orthogonal Hermite polynomials. The then allows us to do the asymptotic analysis, where one of the difficulty is the connection of the various scaling regimes.  
%\\ 

%\todo[inline]{TODO: Higher point SFF (cf. 2$\alpha$-point SFF in \cite[eq (1), (2)]{LPC21})}
%\todo[inline]{TODO: Introduce SFF and DSFF before the model}

\subsection{Elliptic Ginibre ensemble and DSFF}

Let us now describe more precisely the model under consideration. 
We consider the eGinUE (see \cite[Section~2.3]{BF25} and references therein) defined by
\begin{equation} \label{def of eGinibre} 
X:= \frac{\sqrt{1+\tau}}{2}(G+G^*)+\frac{\sqrt{1-\tau}}{2}(G-G^*), \qquad \tau \in [0,1), 
\end{equation} 
where $G$ is a complex Ginibre matrix, i.e. an \(N \times N\) matrix with i.i.d.\ complex Gaussian entries of mean zero and variance \(1/N\). 
Let $\{z_j\}_{j=1}^N$ denote the (complex) eigenvalues of $X$. 
It is well known that, as $N \to \infty$, the empirical spectral measure $\frac{1}{N}\sum_{j=1}^N \delta_{z_j}$ converges in probability to the elliptic law
\begin{equation} \label{def of ellipitc law}
\frac{1}{1-\tau^2} \mathbbm{1}_{ S }(z)\,dA(z), \qquad  S :=  \Big\{ (x,y) \in \R^2: \Big( \frac{x}{1+\tau}\Big)^2 + \Big( \frac{y}{1-\tau} \Big)^2 \le 1 \Big\}.  
\end{equation} 
Here, $dA(z)=d^2z/\pi$ denotes the normalised area measure.

By definition, the $k$-point correlation function $R_{N,k}$ is characterised as follows: for a suitable test function $f:\C^k \to \C$,
\begin{equation} \label{def of RNk test}
\Big\langle \sum_{\substack{1\le i_1,\dots,i_k\le N\\ \text{distinct}}}  
f(z_{i_1},\dots,z_{i_k}) \Big\rangle
=
\int_{\C^k} f(z_1,\dots,z_k)\, R_{N,k}(z_1,\dots,z_k)\, dA(z_1)\cdots dA(z_k),
\end{equation}
where $\langle \cdot \rangle$ denotes the ensemble average.  
For the eGinUE, the eigenvalues form a determinantal point process; see again \cite[Section~2.3]{BF25}. In other words, the $k$-point correlation functions $R_{N,k}$ admit the representation 
\begin{equation} \label{def of RNk det KN}
    R_{N,k}(z_1,\dots,z_k) = \det\big[ K_N(z_j,z_l) \big]_{j,l=1}^k,
\end{equation}
where $K_N$ is a correlation kernel 
%satisfying the 
which is Hermitian\footnote{Such a choice can always be made due to the invariance under multiplying with cocycles.}, 
%symmetry 
$K_N(z,w)=\overline{K_N(w,z)}$; see \eqref{def of KN gen OP} for its representation in terms of the associated orthogonal polynomials. 
In particular, the one- and two-point correlation functions, which will play a central role in our analysis, are given by
\begin{equation} \label{def of RN1 2 KN}
R_{N,1}(z)= K_N(z,z), 
\qquad 
R_{N,2}(z,w)= R_{N,1}(z) R_{N,1}(w) - | K_N(z,w) |^2.
\end{equation}

To describe the DSFF for non-Hermitian random matrices, we introduce a \textit{complex time variable}
\begin{equation} \label{def of complex time variables}
    t+is = Te^{i\theta}, 
\end{equation}  
where $t,s,T \ge 0$ and $\theta \in [0,2\pi)$. 
It will be convenient for our purposes to use both 
%the 
Cartesian coordinates $(t,s)$ and 
%the 
polar coordinates $(T,\theta)$.

Let $z_j = x_j + i y_j$ ($x_j,y_j \in \R$) denote the complex eigenvalues of a non-Hermitian random matrix. 
We define the associated complex Fourier transform by
\begin{equation}
Z_N \equiv Z_N(t,s) := \sum_{j=1}^N e^{ i t x_j + i s y_j }.
\end{equation}
The DSFF $\mathcal{F}_N$ is then defined as
\begin{equation} \label{def of DSFF}
 \mathcal{F}_N  := \big\langle |Z_N|^2 \big\rangle  = \bigg\langle \sum_{j,k=1}^N   e^{ it(x_j-x_k) +is (y_j-y_k) } \bigg\rangle .
\end{equation} 
In principle, general higher powers 
$\big\langle |Z_N|^\nu \big\rangle $, $\nu>0$, 
than quadratic can also be defined; see \cite[Eqs. (1), (2)]{LPC21}. When it is necessary to emphasise the dependence on parameters, we write 
\begin{equation}
 \mathcal{F}_N \equiv   \mathcal{F}_N (t,s) \equiv \mathcal{F}_N(T,\theta) 
\end{equation}
Using \eqref{def of RNk test}, the DSFF can be expressed in terms of two-point statistics. Separating the diagonal and off-diagonal contributions, we have 
\begin{equation} \label{def of DSFF v2}
 \mathcal{F}_N     = N+ \bigg\langle \sum_{1 \le j \not = k \le N}   e^{ it(x_j-x_k) +is (y_j-y_k) } \bigg\rangle = N+ \int_{\C^2} e^{ i t (x - x') + i s (y - y') } R_{N,2}(z,z')\, dA(z)\, dA(z'). 
\end{equation} 

The characteristic dip--ramp--plateau behaviour of the DSFF arises from the interplay between two contributions operating on distinct scales. 
This naturally leads to a decomposition of the DSFF into disconnected and connected parts. 
More precisely, we write 
\begin{equation} \label{def of DSFF sum of disconn conn}
    \mathcal{F}_N  = \mathcal{F}_N^{ \rm (d)} + \mathcal{F}_N^{ \rm (c)} , 
\end{equation}
where
\begin{equation} \label{def of DSFF disc, conn}
 \mathcal{F}_N^{ \rm (d) }  := |\big\langle Z_N \big\rangle|^2, \qquad  \mathcal{F}_N^{ \rm (c) }:= \big\langle |Z_N-\big\langle Z_N \big\rangle|^2 \big\rangle. 
\end{equation}
These are referred to as the \textit{disconnected} and \textit{connected} parts of the DSFF, respectively. 
Using the one- and two-point correlation functions \eqref{def of RN1 2 KN}, these terms can be written more explicitly as
\begin{align} \label{def of FN disconn in terms of RN1 int}
 \mathcal{F}_N^{ \rm (d) } & = \bigg| \sum_{j=1}^N \int_{ \C } e^{it x_j + is y_j} R_{N,1}(z)\,dA(z) \bigg|^2, 
 \\
 \mathcal{F}_N^{ \rm (c) } & =   N + \int_{\C^2} e^{i t (x - x') + i s (y - y')}  \Big( R_{N,2}(z,z') - R_{N,1}(z) R_{N,1}(z') \Big) \,dA(z_1)\,dA(z_2). \label{def of FN conn in terms of RN2 int}
\end{align}
In particular, for the determinantal point process, it follows from \eqref{def of RN1 2 KN} that 
\begin{equation}
 \mathcal{F}_N^{ \rm (c) }=  N - \int_{ \C^2 }  e^{it(x_1-x_2)+is(y_1-y_2)} |K_N(z_1,z_2)|^2 \,dA(z_1)\,dA(z_2) .  
\end{equation}
This representation clarifies the terminology: the disconnected part arises from the product of one-point statistics, while the connected part captures genuine correlations beyond independence. 
We remark that in some parts of the literature, the term $N$ in the above expression for $\mathcal{F}_N^{\rm (c)}$ is separated and referred to as the  \textit{contact} term.
%\todo[inline]{TODO: Eq \eqref{def of DSFF disc, conn} remark on the contact part (relate to the double sum formula), two-point function with delta function}

\subsection{Various non-Hermiticity regimes}

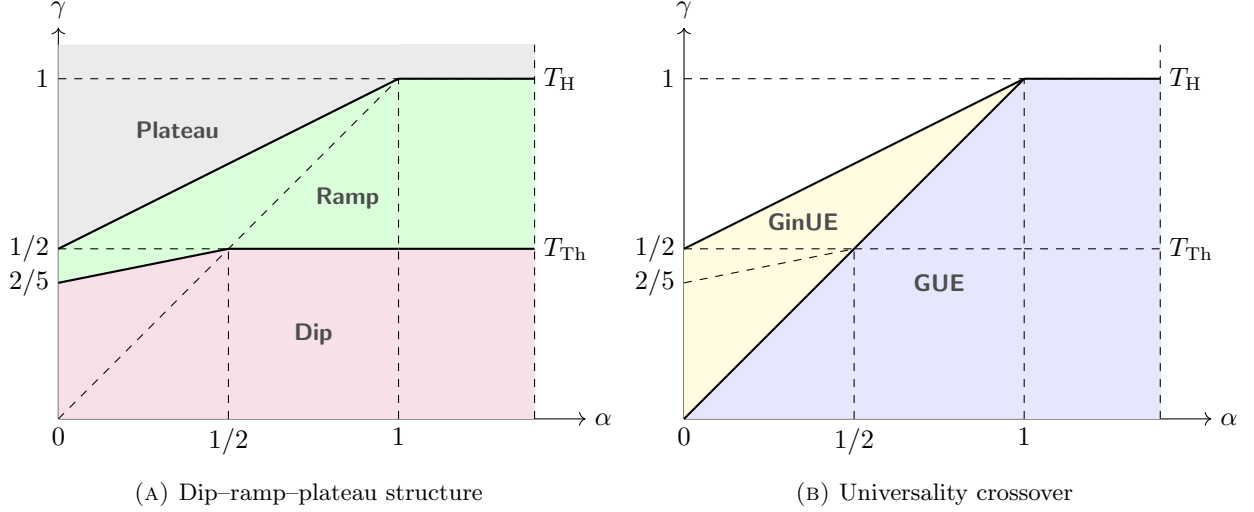
\begin{figure}[t]

\begin{comment} 
\centering
\begin{tikzpicture}[scale=5]

% Axes
\draw[->] (0,0) -- (1.2,0) node[right] {$\alpha$};
\draw[->] (0,0) -- (0,1.1) node[above] {$\gamma$};

% Grid labels
\draw (0,-0.02) node[below] {$0$};
\draw (0.5,-0.02) node[below] {$\tfrac{1}{2}$};
\draw (1,-0.02) node[below] {$1$};

\draw (-0.03,0.4) node[left] {$\tfrac{2}{5}$};
\draw (-0.03,0.5) node[left] {$\tfrac{1}{2}$};
\draw (-0.03,1) node[left] {$1$};

% Lines
\draw[thick] (0,0.5) -- (1,1); % Heisenberg time
\draw[thick] (0,0.4) -- (0.5,0.5); % Thouless time
\draw[thick] (0.5,0.5) -- (1,0.5); % Thouless time
\draw[dashed] (0,1) -- (1,1);    % dashed top line
\draw[dashed] (1,0) -- (1,1);    % dashed right line
\draw[dashed] (0,0.5) -- (0.5,0.5); % dashed midline
\draw[dashed] (0.5,0) -- (0.5,0.5); % dashed midline
\draw[dashed] (0,0) -- (1,1); % dashed diagonal line

% Shaded regions
\fill[pattern=north east lines, pattern color=orange!80] (0,0.5) -- (1,1) -- (0,0) -- cycle;
\fill[pattern=north west lines, pattern color=blue!70] (0,0) -- (1,0) -- (1,1) -- cycle;

% Labels
\node[right] at (1,1) {$T_{\textup{H}}$};
\node[right] at (1,0.5) {$T_{\textup{Th}}$};

\end{tikzpicture}

\end{comment}

\begin{subfigure}{0.48\textwidth}

\begin{tikzpicture}[scale=4.5]
\usetikzlibrary{patterns}

% parameters
\def\xmax{1.4}
\def\ymax{1.1}

% Axes
\draw[->] (0,0) -- (\xmax+0.15,0) node[right] {$\alpha$};
\draw[->] (0,0) -- (0,\ymax+0.05) node[above] {$\gamma$};

% Tick labels (no tick at 1.5)
\node[below] at (0,0) {$0$};
\node[below] at (0.5,0) {$1/2$};
\node[below] at (1,0) {$1$};

\node[left] at (0,0.4) {$2/5$};
\node[left] at (0,0.5) {$1/2$};
\node[left] at (0,1) {$1$};

\fill[fill=purple!12]
  (0,0) -- (0.5,0) -- (0.5,0.5) -- (0,0.4) -- cycle;
% part 2: alpha in [0.5,xmax], under gamma=1/2
\fill[fill=purple!12]
  (0.5,0) rectangle (\xmax,0.5);

% --- Ramp (solid light orange) ---
% part 1: alpha in [0,0.5], between T_Th and T_H
\fill[fill=green!14]
  (0,0.4) -- (0.5,0.5) -- (0.5,0.75) -- (0,0.5) -- cycle;
% part 2: alpha in [0.5,1], between gamma=1/2 and slanted T_H
\fill[fill=green!14]
  (0.5,0.5) -- (1,0.5) -- (1,1) -- (0.5,0.75) -- cycle;
% part 3: alpha in [1,xmax], between gamma=1/2 and gamma=1
\fill[fill=green!14]
  (1,0.5) rectangle (\xmax,1);

% --- Plateau (very light grey) ---
% part 1: alpha in [0,1], above slanted T_H up to ymax
\fill[black!8]
  (0,0.5) -- (1,1) -- (1,\ymax) -- (0,\ymax) -- cycle;
% part 2: alpha in [1,xmax], above gamma=1 up to ymax
\fill[black!8]
  (1,1) rectangle (\xmax,\ymax);

% ------------------------------------------------------------
% Boundary curves
% ------------------------------------------------------------
% T_H
\draw[thick] (0,0.5) -- (1,1);
\draw[thick] (1,1) -- (\xmax,1);

% T_Th
\draw[thick] (0,0.4) -- (0.5,0.5);
\draw[thick] (0.5,0.5) -- (1,0.5);
\draw[thick] (1,0.5) -- (\xmax,0.5);

% Optional dashed guides (similar to your original style)
\draw[dashed] (0,1) -- (\xmax,1);
\draw[dashed] (\xmax,0) -- (\xmax,\ymax);
\draw[dashed] (0,0.5) -- (0.5,0.5);
\draw[dashed] (0.5,0) -- (0.5,0.5);
\draw[dashed] (0,0) -- (1,1);
\draw[dashed] (1,0) -- (1,1);

% Make the boundary at alpha=1 visible (helps the eye)
%\draw[thin] (1,0) -- (1,\ymax);

% ------------------------------------------------------------
% Labels for the curves
% ------------------------------------------------------------
\node[right] at (\xmax,1) {$T_{\textup{H}}$};
\node[right] at (\xmax,0.5) {$T_{\textup{Th}}$};

\node[black!70] at (0.75,0.25) {\small\textsf{\textbf{Dip}}};
\node[black!70] at (0.85,0.65) {\small\textsf{\textbf{Ramp}}};
\node[black!70] at (0.35,0.85) {\small\textsf{\textbf{Plateau}}};
\end{tikzpicture}
\subcaption{Dip--ramp--plateau structure}
	\end{subfigure}	
		\begin{subfigure}{0.48\textwidth}
\begin{tikzpicture}[scale=4.5]
\usetikzlibrary{patterns}

\def\xmax{1.4}
\def\ymax{1.1}

\draw[->] (0,0) -- (\xmax+0.15,0) node[right] {$\alpha$};
\draw[->] (0,0) -- (0,\ymax+0.05) node[above] {$\gamma$};

\node[below] at (0,0) {$0$};
\node[below] at (0.5,0) {$1/2$};
\node[below] at (1,0) {$1$};

\node[left] at (0,0.4) {$2/5$};
\node[left] at (0,0.5) {$1/2$};
\node[left] at (0,1) {$1$};

\fill[fill=yellow!14!]
  (0,0.5) -- (1,1) -- (0,0) -- cycle;

% GUE
\fill[fill=blue!10!]
  (0,0) -- (1,0) -- (1,1) -- cycle;

\fill[fill=blue!10!]
  (1,0) rectangle (\xmax,1);

% T_H
\draw[thick] (0,0.5) -- (1,1);
\draw[thick] (1,1) -- (\xmax,1);

% T_Th
\draw[dashed] (0,0.4) -- (0.5,0.5);
\draw[dashed] (0.5,0.5) -- (1,0.5);
\draw[dashed] (1,0.5) -- (\xmax,0.5);
 
\draw[dashed] (0,1) -- (\xmax,1);
\draw[dashed] (\xmax,0) -- (\xmax,\ymax);
\draw[dashed] (0,0.5) -- (0.5,0.5);
\draw[dashed] (0.5,0) -- (0.5,0.5);
\draw[thick] (0,0) -- (1,1);
\draw[dashed] (1,0) -- (1,1);

\node[right] at (\xmax,1) {$T_{\textup{H}}$};
\node[right] at (\xmax,0.5) {$T_{\textup{Th}}$};

\node[black!70] at (0.75,0.4) {\small\textsf{\textbf{GUE}}};
\node[black!70] at (0.35,0.58) {\small\textsf{\textbf{GinUE}}}; 
\end{tikzpicture}
\subcaption{Universality crossover}
	\end{subfigure}

\caption{Phase diagrams of the DSFF in the $(\alpha,\gamma)$--plane, where the non-Hermiticity and time scalings are given by $1-\tau=N^{-\alpha}$ and $T=N^{\gamma}\mathsf{T}$, with $\mathsf{T}>0$ fixed. 
The Thouless time $T_{\mathrm{Th}}$ and the 
generalised 
Heisenberg time $T_{\mathrm{H}}$ are indicated on the figure. 
\textup{(A)} Dip--ramp--plateau structure of the DSFF across different time scales. 
\textup{(B)} Universality crossover of the DSFF, showing GUE-type behaviour for $\gamma\le\alpha$ and GinUE-type behaviour for $\alpha<\gamma<\frac{1+\alpha}{2}$. This agrees with the regimes identified in \cite[Figure 1]{ADM24} for the number variance.}
\label{fig:alpha_gamma}
\end{figure}

In analysing the asymptotic behaviour of the DSFF for the eGinUE, it is necessary to consider several scaling regimes with respect to the system size $N$. 
In particular, two parameters must be scaled appropriately: the complex time parameter introduced in \eqref{def of complex time variables}, and the non-Hermiticity parameter $\tau$. 

The statistical behaviour of the DSFF under the various scalings constitutes one of the main 
%observables 
results of this work. 
A summary of the resulting regimes is presented in Figure~\ref{fig:alpha_gamma}, and will be discussed in detail in the next section.

From the elliptic law \eqref{def of ellipitc law}, one observes that the typical spacing between eigenvalues scales as
\begin{equation} \label{def of typical ev scaling}
\begin{cases}
\sqrt{ \frac{1-\tau^2}{N} } &\textup{if } 1-\tau \gg \frac{1}{N},
\smallskip 
\\
\frac{1}{N} &\textup{if } 1-\tau= O(\frac{1}{N}).
\end{cases}
\end{equation}
This reflects a crossover from two-dimensional bulk behaviour to an effectively one-dimensional regime as $\tau \to 1$. 

In the study of the eGinUE, it is essential to distinguish different regimes of non-Hermiticity according to the scaling of the parameter $\tau$ with the matrix size $N$. Equivalently, these regimes are characterised by the typical fluctuation scale of the imaginary parts of the eigenvalues. We consider the general scaling
\begin{equation} \label{def of tau scaling general}
\tau=  1- \frac{\kappa}{N^{\alpha}}, \qquad \kappa>0, \, \alpha \ge 0. 
\end{equation}
Depending on the value of $\alpha$, the following regimes arise.
 
\begin{itemize}
    \item \textbf{Strong non-Hermiticity.} 
    The parameter $\tau \in [0,1)$ is kept fixed as $N \to \infty$, corresponding to $\alpha=0$ in \eqref{def of tau scaling general}. In this regime, the eigenvalues form a genuinely two-dimensional cloud in the complex plane, and non-Hermitian effects persist 
    %at 
    on the macroscopic scale.
    \smallskip
    \item \textbf{Mesoscopic non-Hermiticity.} This intermediate regime corresponds to $\alpha \in (0,1)$ in \eqref{def of tau scaling general}. The 
    %imaginary 
    fluctuations of the imaginary parts occur on a mesoscopic scale, interpolating between the strong and weak non-Hermiticity regimes. We refer to~\cite{ADM24} for a recent analysis in this setting.
    \smallskip
    \item \textbf{Weak non-Hermiticity.}  
    The parameter $\tau$ approaches the Hermitian limit on the scale
    $\alpha=1$ in \eqref{def of tau scaling general}. In this regime, the eigenvalues concentrate near the real
    axis, while non-Hermitian effects persist 
    %at 
    on the microscopic scale.
    Consequently, the local statistics interpolate between those of
    %Hermitian ensembles and the Ginibre ensemble
    the GUE and the GinUE; 
    see~\cite{FKS97,FKS98}. More generally, one may consider $\alpha \ge 1$. The case $\alpha=1$ is critical, in the sense that the typical fluctuations of the
    imaginary parts of the eigenvalues are of the same order as the mean
    eigenvalue spacing. When $\alpha>1$, the model lies closer to the
    Hermitian regime, as the 
    %imaginary 
    fluctuations of the imaginary parts become negligible
    compared to the eigenvalue spacing.
\end{itemize}

We note that the regimes of non-Hermiticity discussed above are characterised in terms of bulk statistics. 
In contrast, for critical edge statistics, the relevant scaling is governed by the critical exponent $\alpha = 1/3$ in the microscopic limit; see \cite{Be10,GNV02,AB23}. 

Notice that the reciprocal of the typical eigenvalue spacing \eqref{def of typical ev scaling} scales as
\begin{equation}
\begin{cases} \label{def of reciprocal spacing}
N^{1/2} &\textup{for strong non-Hermiticity},
\smallskip 
\\
N^{ (1+\alpha)/2 }  &\textup{for mesoscopic non-Hermiticity},
\smallskip 
\\
N &\textup{for weak non-Hermiticity}.
\end{cases}
\end{equation}
As we will see below, this quantity naturally identifies the corresponding 
generalised \emph{Heisenberg time} $T_\textup{H}$ in each regime
as the scale, on which universal microscopic fluctuations set in.

\section{Main results}

In this section, we present our main results. 
We study the large-$N$ asymptotics of the DSFF across different regimes of non-Hermiticity and over a range of time scales. 
In all cases, we observe a transition in the dominant contribution between the disconnected and connected parts of the DSFF.

More precisely, we analyse the DSFF on general time scales of the form
\begin{equation} \label{def of time scaling sf T t}
T = N^{\gamma} \, \mathsf{T}, \qquad t= N^{\gamma}\, \mathsf{t}, \qquad \gamma \ge 0,
\end{equation}
where $\mathsf{T}, \mathsf{t}>0$ are independent of $N$. 
Within this framework, we identify the \emph{Thouless time} $T_{\mathrm{Th}}$, at which the DSFF exhibits the dip behaviour, and the 
%\emph
generalised 
{Heisenberg time} $T_{\mathrm{H}}$, at which the DSFF reaches its plateau.
Here, $T_{\mathrm{Th}}$ is the time scale needed to explore the full system size, when the effects of quantum localisation will cutoff classical diffusion \cite{Ha10}.

We recall that the Bessel function of the first kind is defined by 
\begin{equation} \label{def of Bessel J_nu}
J_\nu(z)= \sum_{k=0}^\infty \frac{ (-1)^k (z/2)^{2k+\nu} }{ k!\,\Gamma(\nu+k+1) };
\end{equation}
see e.g. \cite[Chapter 10]{NIST}. 
Our first result concerns the regime of strong non-Hermiticity, in which the parameter $\tau \in [0,1)$ is held fixed, 
corresponding to $\alpha=0$.
Recall that the disconnected component $\mathcal{F}_{N}^{(\textup{d})}$ and the connected component $\mathcal{F}_{N}^{(\textup{c})}$ of the DSFF are defined in \eqref{def of DSFF disc, conn}.

In the following theorem, we introduce the superscripts $(\mathrm{d})$ and $(\mathrm{c})$ to denote the disconnected and connected parts, respectively, and the subscript ``$\mathrm{s}$'' to indicate the strong non-Hermiticity regime. 

\begin{thm}[\textbf{Asymptotics of the disconnected/connected DSFF at strong non-Hermitcity}] \label{thm strong nH}
Let $\tau \in [0,1)$ be fixed 
corresponding to $\alpha=0$, 
and write 
\begin{equation} \label{def of eta}
\eta \equiv \eta(\tau,\theta) := \frac{\cos\theta(1+\tau)}{2}+\frac{i\sin\theta(1-\tau)}{2}. 
\end{equation}
Then we have the following. 
\begin{itemize}
    \item[(i)] \textup{\textbf{(Dip-ramp regime)}} Let $\gamma\in[0,1/2]$. Then as $N \to \infty$, we have 
    \begin{align} \label{asymp of FN d/c snH} \begin{split} 
    \mathcal{F}_{N}^{(\textup{d})}(N^{\gamma}\mathsf{T},\theta) &= N^{2-3\gamma} \Big(\mathcal{F}_{\textup{s}}^{(\textup{d})}(\mathsf{T},\theta) +O(N^{-\epsilon_1})\Big),
    \\
    \mathcal{F}_{N}^{(\textup{c})}(N^{\gamma}\mathsf{T},\theta) & = N^{2\gamma} \Big(\mathcal{F}_{\textup{s}}^{(\textup{c})}(\mathsf{T},\theta) +O(N^{-\epsilon_2})\Big),
    \end{split}
    \end{align} 
    where  
    \begin{equation} \label{strong nH dis}
    \mathcal{F}_{\textup{s}}^{(\textup{d})}(\mathsf{T},\theta) := \begin{cases} \displaystyle
        \frac{J_{1} (2\vert\eta \mathsf{T}\vert)^2}{\vert\eta \mathsf{T}\vert^2} & \textup{for }  \gamma=0, 
        \smallskip 
        \\
        \displaystyle
        \frac{1}{2\pi |\eta \mathsf{T}|^3} \Big(1-\sin(4N^{\gamma}|\eta \mathsf{T}|)\Big) & \textup{for } \gamma\in(0,\frac12), 
        \smallskip 
        \\ 
        \displaystyle
        \exp\Big(-\frac{1-\tau^2}{4}\mathsf{T}^2\Big) \frac{1}{2\pi |\eta \mathsf{T}|^3} \Big(1-\sin(4N^{\gamma}|\eta \mathsf{T}|)\Big) & \textup{for } \gamma=\frac12  
    \end{cases}
\end{equation}
and
\begin{equation} \label{strong nH conn}
    \mathcal{F}_{\textup{s}}^{(\textup{c})}(\mathsf{T},\theta) := \begin{cases} \displaystyle
        \frac{1-\tau^2}{4}\mathsf{T}^2 + |\eta\mathsf{T}|^2\Big( 2J_0(2\vert\eta \mathsf{T}\vert)^2 + 2J_1(2\vert\eta \mathsf{T}\vert)^2 - \frac{J_0(2\vert\eta \mathsf{T}\vert)J_1(2\vert\eta \mathsf{T}\vert)}{|\eta\mathsf{T}|} \Big) & \textup{for } \gamma=0, 
        \smallskip 
        \\
        \displaystyle
        \frac{1-\tau^2}{4} \mathsf{T}^2 & \textup{for } \gamma\in(0,\frac12), 
        \smallskip 
        \\ \displaystyle
        1-\exp\Big(-\frac{1-\tau^2}{4}\mathsf{T}^2\Big) & \textup{for } \gamma=\frac12.
    \end{cases}
\end{equation}
In \eqref{asymp of FN d/c snH}, the constants $\epsilon_1 , \epsilon_2$ are given by
\begin{equation} \label{def of epsilon 1,2}
    \epsilon_1 = \begin{cases} \displaystyle
        2-3\gamma & \textup{for } \gamma\in[0,\frac25), 
        \smallskip 
        \\ \displaystyle
        2\gamma & \textup{for } \gamma\in[\frac25,\frac12], 
    \end{cases} \qquad
    \epsilon_2 = \begin{cases} \displaystyle
        1 & \textup{for } \gamma=0,
        \smallskip 
        \\ \displaystyle
        \gamma & \textup{for } \gamma\in(0,\frac13), 
        \smallskip 
        \\ \displaystyle
        1-2\gamma & \textup{for } \gamma\in[\frac13,\frac12), 
        \smallskip 
        \\ \displaystyle
        1/2 & \textup{for } \gamma=\frac12. 
    \end{cases}
\end{equation} %\eqref{def of epsilon 1,2}. 
%\smallskip 
\item[(ii)] \textup{\textbf{(Plateau regime)}} Let $\gamma > 1/2$. Then as $N \to \infty$, we have 
\begin{equation} \label{asymp of FN p snH}
\begin{split}
    & \mathcal{F}_{N}^{(\textup{d})}(N^{\gamma}\mathsf{T},\theta) = \exp\Big(-N^{2\gamma-1} \Phi_{\textup{s}}(\mathsf{T},\theta)+O(\mathsf{e}_1(N))\Big), \\
    & \mathcal{F}_{N}^{(\textup{c})}(N^{\gamma}\mathsf{T},\theta) = N-\exp\Big(-N^{2\gamma-1} \Phi_{\textup{s}}(\mathsf{T},\theta)+O(\mathsf{e}_1(N))\Big),
\end{split}
\end{equation}
where
\begin{align}
\Phi_{\textup{s}}(\mathsf{T},\theta) & := \begin{cases} \displaystyle
    \frac{1-\tau^2}{4}\mathsf{T}^2 &\textup{for } \gamma\in (\frac12,1) \textup{ or } \gamma=1, \, |\eta\mathsf{T}|<2, 
    \smallskip 
    \\  
    \displaystyle
    \frac{1-\tau^2}{4}\mathsf{T}^2+|\eta\mathsf{T}|\sqrt{|\eta\mathsf{T}|^2-4}-4\,\textup{arccosh}\frac{|\eta\mathsf{T}|}{2} &\textup{for } \gamma=1, \, |\eta\mathsf{T}|>2, 
    \smallskip 
    \\  
    \displaystyle
    \frac{1-\tau^2}{4}\mathsf{T}^2+|\eta \mathsf{T}|^2 &\textup{for } \gamma>1.
    \end{cases} 
\end{align}
Here, the error term $\mathsf{e}_1(N)$ is given by \begin{equation} \label{def of e1}
    \mathsf{e}_1(N) = \begin{cases} \displaystyle
        \log N & \textup{for } \gamma\in(\frac12,1], \smallskip \\ \displaystyle
        N \log N & \textup{for } \gamma>1.
    \end{cases}
\end{equation} 
\end{itemize}  
\end{thm}

\begin{rem}[Dependence on the angular parameter $\theta$]
Note that in Theorem~\ref{thm strong nH}, the resulting asymptotic formulas may depend solely on $\mathsf{T}$, or may additionally involve the angular parameter $\theta$. In the latter case, the dependence on $\theta$ enters through the parametrisation $\eta$ defined in \eqref{def of eta}; that is, $\theta$ always appears together with $\tau$ in the combination specified in \eqref{def of eta}. In particular, it enters only through the quantity $|\eta \mathsf{T}|$.  
Moreover, when $\tau = 0$, one has $|\eta| = \tfrac{1}{2}$, and hence the resulting formulas are independent of the angular parameter $\theta$.
\end{rem}

By \eqref{def of DSFF sum of disconn conn}, the following result is an immediate consequence of Theorem~\ref{thm strong nH}. 

\begin{cor}[\textbf{Leading order asymptotics of the DSFF at strong non-Hermitcity}] \label{Cor_DSFF strong}  Let $\tau \in [0,1)$ be fixed,
corresponding to $\alpha=0$. Then as $N \to \infty$, we have 
\begin{equation}
\mathcal{F}_N(N^\gamma \mathsf{T},\theta) \sim 
\begin{cases}
N^{2-3\gamma} \mathcal{F}_{ \rm s }^{\rm (d)} (\mathsf{T},\theta) &\textup{for } \gamma \in [0,\frac25), 
\smallskip 
\\
N^{ \frac45 } \big( \mathcal{F}_{ \rm s }^{\rm (d)} (\mathsf{T},\theta)+  \mathcal{F}_{ \rm s }^{\rm (c)} (\mathsf{T},\theta) \big)  &\textup{for } \gamma = \frac25, 
\smallskip 
\\
N^{2\gamma}  \mathcal{F}_{ \rm s }^{\rm (c)} (\mathsf{T},\theta)  &\textup{for } \gamma \in (\frac25,\frac12], 
\smallskip 
\\
N &\textup{for } \gamma >\frac12. 
\end{cases}
\end{equation}
\end{cor}

%See Figure~\ref{Fig_DSFF strong} for the numerics. 

In Corollary~\ref{Cor_DSFF strong}, we present only the leading-order asymptotic behaviour of the DSFF. In particular, one observes that when $\gamma < 2/5$, the disconnected part is dominant, whereas for $\gamma > 2/5$, the connected part provides the leading contribution. Furthermore, when $\gamma > 1/2$, the DSFF exhibits a flat (plateau) behaviour
$\sim N$, which is constant for fixed $N$; 
indeed, in Theorem~\ref{thm strong nH}, the error term in this regime is exponentially small.

See Figure~\ref{Fig_DSFF strong} for 
illustration, presenting our numerics.

We also note that for $\gamma \in (2/5, 1/2]$, where the connected part dominates, its behaviour is quadratic by \eqref{strong nH conn}. This quadratic ramp is in contrast to the linear ramp behaviour that arises in the SFF of the GUE.

In the regime of strong non-Hermiticity, part of our result in Theorem~\ref{thm strong nH} was previously obtained in~\cite[Eqs.~(10), (11)]{SKSK24}. Building on earlier work~\cite{Fo21}, the authors analysed the cases $\gamma = 0$ for the disconnected part and $\gamma = 1$ for the connected part. In~\cite{SKSK24}, this was further extended to the intermediate regime $0 < \gamma < 1$, where the Thouless time $T_{\textup{Th}} = O(N^{2/5})$ was proposed. Moreover, in the weak non-Hermiticity regime $\tau = 1 - O(N^{-1})$, they obtained the corresponding Thouless time $T_{\textup{Th}} = O(N^{1/2})$. Our results provide a rigorous confirmation of this scaling and give a complete description of the DSFF across all values of $\gamma$ in the strong non-Hermiticity regime.

\begin{figure}[t]
\centering
\begin{tikzpicture}[
    fig/.style={
        draw=none,
        inner sep=0pt,
        outer sep=0pt
    }
]

% --- top figure ---
\node[fig] (fig1) at (0,0)
    {\includegraphics[width=0.52\textwidth]{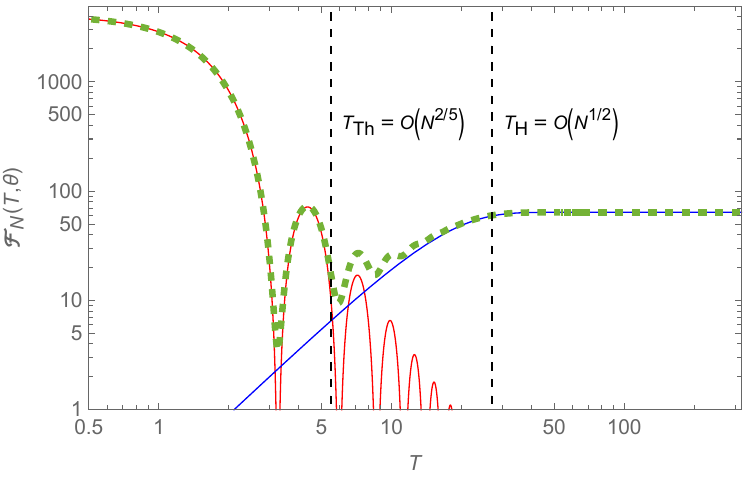}};
% \node[above=2mm of fig1] {\large (A)};

% --- bottom left figure ---
\node[fig] (fig2) at (-4.6,-6.5)
    {\includegraphics[width=0.42\textwidth]{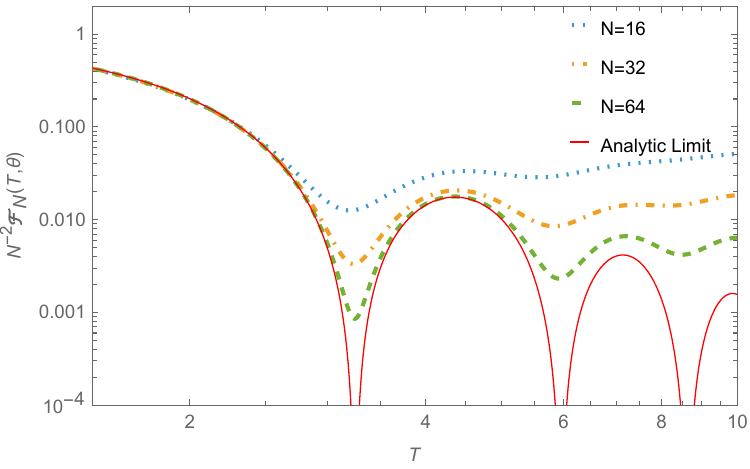}};
% \node[above=2mm of fig2] {\large (B)};

% --- bottom right figure ---
\node[fig] (fig3) at (4.6,-6.5)
    {\includegraphics[width=0.42\textwidth]{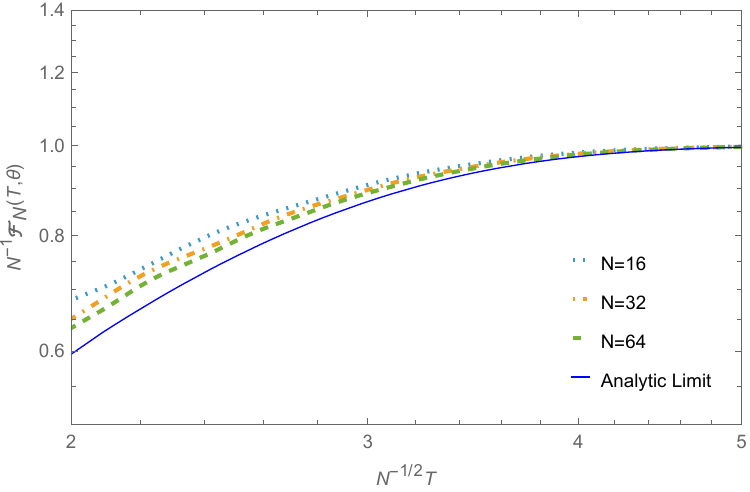}};
% \node[above=2mm of fig3] {\large (C)};

% --- connecting lines ---
\draw[thick, gray, dashed] ($(fig1.south west)+(1.4,0.6)$) -- ($(fig2.north west)+(0.9,0.1)$);
\draw[thick, gray, dashed] ($(fig1.south west)+(4.2,0.6)$) -- ($(fig2.north east)+(-0.1,0.1)$);
\draw[thick, gray, dashed] ($(fig1.south)+(0.7,0.6)$) -- ($(fig3.north west)+(0.6,0.1)$);
\draw[thick, gray, dashed] ($(fig1.south)+(2.3,0.6)$) -- ($(fig3.north east)+(-0.1,0.1)$); 
\end{tikzpicture}
\caption{DSFF in the regime of strong non-Hermiticity 
$(\alpha=0)$ 
with $\tau = 0.3$ and $\theta = \pi/6$. Here, all axes are on a logarithmic scale. 
Top: DSFF for $N = 64$ (green, dashed), together with the analytic limits for $\gamma = 0$ (red, solid) and $\gamma = \tfrac{1}{2}$ (blue, solid). 
Bottom left: DSFF for $\gamma = 0$ for varying $N$, compared with the analytic limit. 
Bottom right: DSFF for $\gamma = \tfrac{1}{2}$ for varying $N$, compared with the analytic limit.} \label{Fig_DSFF strong}
\end{figure}
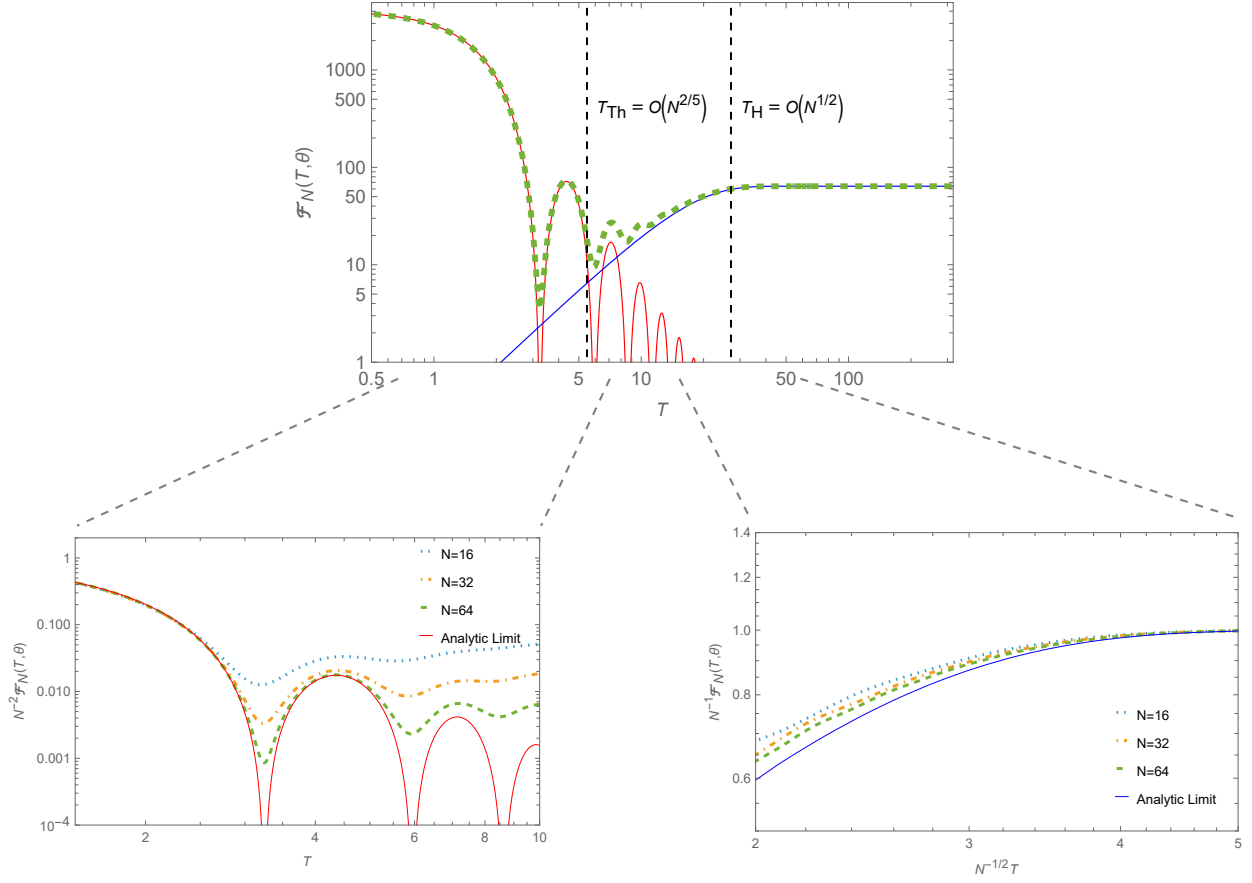

\medskip 

Our next result concerns the mesoscopic scaling regime, in which $\tau \to 1$. In this setting, rather than using the parametrisation $\eta$ in \eqref{def of eta} that captures the angular contribution at strong non-Hermiticity, it is more convenient to work with both polar and Cartesian coordinates as introduced in \eqref{def of time scaling sf T t}. In particular, by \eqref{def of complex time variables}, we have
\begin{equation}
\mathsf{t} = \mathsf{T} \cos\theta.
\end{equation}
Indeed, recalling that in Theorem~\ref{thm strong nH}, the $\theta$-dependence of the resulting asymptotic formulas enters through the combination $|\eta \mathsf{T}|$, one may interpret the parameter $\mathsf{t}$ as $|\eta \mathsf{T}|$. In the case $\tau = 1$, this reduces to $\eta = \cos \theta$.

Accordingly, in the statement of the following theorem, we describe the asymptotic behaviour in terms of the variables $\mathsf{T}$ and $\mathsf{t}$. 
We then obtain the following result. As before, we introduce the subscript ``$\mathrm{m}$'' to indicate the mesoscopic non-Hermiticity regime. 

\begin{thm}[\textbf{Asymptotics of the disconnected/connected DSFF at mesoscopic non-Hermitcity}] \label{thm mesoscopic nH}
Let $\tau$ be scaled as \eqref{def of tau scaling general} with $\alpha \in (0,1).$

\begin{itemize}
    \item[(i)] \textup{\textbf{(Dip-ramp regime)}} Let $\gamma\in[0,(1+\alpha)/2]$. Then as $N \to \infty$, we have 
    \begin{align} \label{asymp of FN d/c mnH}
    \begin{split}
\mathcal{F}_{N}^{(\textup{d})}(N^{\gamma}\mathsf{T},\theta) &= N^{2-3\gamma} \Big(\mathcal{F}_{\textup{m}}^{(\textup{d})}(\mathsf{T},\theta) + O(N^{-\epsilon_3})\Big),
    \\
 \mathcal{F}_{N}^{(\textup{c})}(N^{\gamma}\mathsf{T},\theta) &= N^{\max\{2\gamma-\alpha,\gamma\}} \Big(\mathcal{F}_{\rm m}^{(\textup{c})}(\mathsf{T},\theta) + O(N^{-\epsilon_4})\Big), 
    \end{split}
    \end{align} 
where
\begin{equation} \label{mesoscopic nH dis}
    \mathcal{F}_{ \textup{m} }^{(\textup{d})}(\mathsf{T},\theta) := \begin{cases} \displaystyle
        \frac{J_{1} (2\mathsf{t})^2}{\mathsf{t}^2} & \textup{for } \gamma=0,  
        \smallskip 
        \\
        \displaystyle
        \frac{1}{2\pi \mathsf{t}^3} \Big(1-\sin(4N^{\gamma}\mathsf{t})\Big) & \textup{for } \gamma\in (0,\alpha),  
        \smallskip 
        \\
        \displaystyle
        \frac{1}{2\pi \mathsf{t}^3} \Big(1-\sin\big(4N^{\gamma}\mathsf{t}+O(N^{\gamma-\alpha})\big)\Big) & \textup{for } \gamma\in [\alpha,\frac{1+\alpha}{2}),  
        \smallskip 
        \\
        \displaystyle
        \exp\Big(-\frac{\kappa}{2}\mathsf{T}^2\Big) \frac{1}{2\pi \mathsf{t}^3} \Big(1-\sin\big(4N^{(1+\alpha)/2}\mathsf{t}+O(N^{(1-\alpha)/2}\big)\Big) & \textup{for } \gamma=\frac{1+\alpha}{2}
    \end{cases}
\end{equation}
and 
\begin{equation} \label{mesoscopic nH conn}
    \mathcal{F}_{ \textup{m} }^{(\textup{c})}(\mathsf{T},\theta) := \begin{cases} \displaystyle
        \mathsf{t}^2\Big(2J_0(2\mathsf{t})^2 + 2J_1(2\mathsf{t})^2 - \frac{J_0(2\mathsf{t})J_1(2\mathsf{t})}{\mathsf{t}}\Big) & \textup{for } \gamma=0 ,
        \smallskip 
        \\
      \displaystyle   \frac{2}{\pi}\mathsf{t}  & \textup{for } \gamma\in (0,\alpha) ,
        \smallskip 
        \\
        \displaystyle  \frac{2}{\pi}\mathsf{t}+\frac{\kappa}{2}\mathsf{T}^2  & \textup{for } \gamma=\alpha ,
        \smallskip 
        \\
        \displaystyle  \frac{\kappa}{2}\mathsf{T}^2  & \textup{for } \gamma\in (\alpha,\frac{1+\alpha}{2}) ,
        \smallskip 
        \\
        \displaystyle
        1-\exp\Big(-\frac{\kappa}{2}\mathsf{T}^2\Big) & \textup{for } \gamma=\frac{1+\alpha}{2}. 
    \end{cases}
\end{equation}
In \eqref{asymp of FN d/c mnH}, the constants $\epsilon_3 , \epsilon_4 >0$ are given by 
\begin{equation} \label{def of epsilon 3,4}
    \epsilon_3 = \begin{cases} \displaystyle
        \alpha-\gamma & \textup{for } \gamma\in[0,\alpha), 
        \smallskip 
        \\ \displaystyle
        \alpha & \textup{for } \gamma\in[\alpha,\frac{1+\alpha}{2}], 
    \end{cases} \qquad
    \epsilon_4 = \begin{cases} \displaystyle
        \alpha & \textup{for } \gamma=0,
        \smallskip 
        \\ \displaystyle
        \min\{2\gamma,\alpha-\gamma\} & \textup{for } \gamma\in(0,\alpha),
        \smallskip 
        \\ \displaystyle
        \min\{\alpha,1-\alpha\} & \textup{for } \gamma=\alpha,
        \smallskip 
        \\ \displaystyle
        \min\{\alpha,\gamma-\alpha,1+\alpha-2\gamma\} & \textup{for } \gamma\in(\alpha,\frac{1+\alpha}{2}),
        \smallskip 
        \\ \displaystyle
        \min\Big\{\alpha,\frac{1-\alpha}{2}\Big\} & \textup{for } \gamma=\frac{1+\alpha}{2} 
    \end{cases}
\end{equation} %\smallskip 
\item[(ii)] \textup{\textbf{(Plateau regime)}} Let $\gamma>(1+\alpha)/2$. Then as $N \to \infty$, we have 
\begin{equation} \label{asymp of FN p mnH}
\begin{split}
    & \mathcal{F}_{N}^{(\textup{d})}(N^{\gamma}\mathsf{T},\theta) = \begin{cases} \displaystyle
        \exp\Big(-N^{2\gamma-1-\alpha} \Phi_{\textup{m}}(\mathsf{T},\theta)+O(\mathsf{e}_2(N))\Big) &\textup{for } \gamma\in (\frac12,1) \textup{ or } \gamma=1, \, \mathsf{t}<2, \smallskip \\  \displaystyle
        \exp\Big(-N^{2\gamma-1} \Phi_{\textup{m}}(\mathsf{T},\theta)+O(\mathsf{e}_2(N))\Big) &\textup{for } \gamma=1, \, \mathsf{t}>2 \textup{ or } \gamma>1, 
    \end{cases} \\
    & \mathcal{F}_{N}^{(\textup{c})}(N^{\gamma}\mathsf{T},\theta) = N-\begin{cases} \displaystyle
        \exp\Big(-N^{2\gamma-1-\alpha} \Phi_{\textup{m}}(\mathsf{T},\theta)+O(\mathsf{e}_2(N))\Big) &\textup{for } \gamma\in (\frac12,1) \textup{ or } \gamma=1, \, \mathsf{t}<2, \smallskip \\  \displaystyle
        \exp\Big(-N^{2\gamma-1} \Phi_{\textup{m}}(\mathsf{T},\theta)+O(\mathsf{e}_2(N))\Big) &\textup{for } \gamma=1, \, \mathsf{t}>2 \textup{ or } \gamma>1, 
    \end{cases}
\end{split}
\end{equation}
where
\begin{equation}
\Phi_{\textup{m}}(\mathsf{T},\theta) := \begin{cases} \displaystyle
    \frac{\kappa}{2}\mathsf{T}^2 & \textup{for } \gamma\in(\frac{1+\alpha}{2},1) \textup{ or } \gamma=1, \, \mathsf{t}<2,
    \smallskip
    \\
    \displaystyle
    \mathsf{t}\sqrt{\mathsf{t}^2-4}-4\,\textup{arccosh}\frac{\mathsf{t}}{2} & \textup{for } \gamma=1, \, \mathsf{t}>2,
    \smallskip
    \\
    \displaystyle
    \mathsf{t}^2 & \textup{for } \gamma>1.
    \end{cases}
\end{equation}
Here, the error term $\mathsf{e}_2(N)$ is given by \begin{equation} \label{def of e2}
    \mathsf{e}_2(N) = \begin{cases} \displaystyle
        \log N &\textup{for } \gamma\in(\frac{1+\alpha}{2},1) \textup{ or } \gamma=1, \, \mathsf{t}<2, \smallskip \\ \displaystyle
        N^{1-\alpha} &\textup{for } \gamma=1, \, \mathsf{t}>2, \smallskip \\ \displaystyle
        N^{2\gamma-1-\alpha} &\textup{for } \gamma>1.
    \end{cases}
\end{equation} 
\end{itemize}
\end{thm}

As a consequence of Theorem~\ref{thm mesoscopic nH}, by \eqref{def of DSFF sum of disconn conn}, we have the following.

\begin{cor}[\textbf{Leading order asymptotics of the DSFF at mesoscopic non-Hermitcity}]  \label{Cor_DSFF mesoscopic}  Let $\tau$ be scaled as \eqref{def of tau scaling general} with $\alpha \in (0,1).$
As $N \to \infty$, we have 
\begin{equation}
\mathcal{F}_N(N^\gamma \mathsf{T},\theta) \sim 
\begin{cases}
N^{2-3\gamma} \mathcal{F}_{ \rm m }^{\rm (d)} (\mathsf{T},\theta) &\textup{for } \gamma \in [0, \min\{\frac12, \frac{2+\alpha}{5} \} ), 
\smallskip 
\\
N^{ \max\{ \frac12, \frac{4-\alpha}{5} \}  } \big( \mathcal{F}_{ \rm m }^{\rm (d)} (\mathsf{T},\theta)+  \mathcal{F}_{ \rm m }^{\rm (c)} (\mathsf{T},\theta) \big)  &\textup{for } \gamma = \min\{\frac12, \frac{2+\alpha}{5} \}, 
\smallskip 
\\
N^{ \max\{ 2\gamma-\alpha,\gamma  \} }  \mathcal{F}_{ \rm m }^{\rm (c)} (\mathsf{T},\theta)  &\textup{for } \gamma \in (\min\{\frac12, \frac{2+\alpha}{5} \} , \frac{1+\alpha}{2} ], 
\smallskip 
\\
N &\textup{for } \gamma >\frac{1+\alpha}{2}. 
\end{cases}
\end{equation}
\end{cor}

% \todo[inline]{TODO; leading order asymptotic}

As before, we present the leading-order asymptotics of the DSFF. In the mesoscopic regime, the dominant behaviour exhibits a similar transition: for small values of $\gamma$, the disconnected part is dominant, whereas for larger values of $\gamma$, the connected part provides the leading contribution. Here, the precise transition threshold depends on the mesoscopic scaling parameter $\alpha$.

\begin{rem}[Behaviour in the ramp regime] 
As discussed above, the mesoscopic scaling regime exhibits a critical interpolation of the DSFF, in which both linear and quadratic behaviours arise depending on the value of $\gamma$; see \eqref{mesoscopic nH conn} and Figure~\ref{Fig_DSFF mesocopic}. The quadratic behaviour is characteristic of strongly non-Hermitian ensembles (GinUE-type), whereas the linear behaviour reflects Hermitian (GUE-type) statistics.

More precisely, for $\gamma \in (0,\alpha)$ the statistics are GUE-like, while for $\gamma \in (\alpha, \tfrac{1+\alpha}{2})$ they are GinUE-like; see Figure~\ref{fig:alpha_gamma} (B). As $\alpha \to 0$, corresponding to strong non-Hermiticity, the interval $(0,\alpha)$ shrinks and the linear behaviour becomes negligible. Conversely, as $\alpha \to 1$, approaching the Hermitian limit, the interval $(\alpha, \tfrac{1+\alpha}{2})$ shrinks and the quadratic behaviour is suppressed.

In the ramp regime, more refined distinctions appear. On the one hand, when
\[
\gamma \in \Bigl(\min\{\tfrac{1}{2}, \tfrac{2+\alpha}{5}\}, \tfrac{1+\alpha}{2}\Bigr] \cap (0,\alpha),
\]
the connected contribution exhibits asymptotically linear growth, in agreement with the GUE ramp. This regime is non-empty only for $\alpha > \tfrac{1}{2}$. On the other hand, when
\[
\gamma \in \Bigl(\min\{\tfrac{1}{2}, \tfrac{2+\alpha}{5}\}, \tfrac{1+\alpha}{2}\Bigr] \cap \Bigl(\alpha, \tfrac{1+\alpha}{2}\Bigr),
\]
which is non-empty for all $\alpha \in (0,1)$, the connected part displays quadratic growth, characteristic of the GinUE ramp.

We expect this dichotomy between GUE- and GinUE-type behaviour to be a universal feature of macroscopic observables (such as the DSFF) for the eGinUE. Similar transitions have been observed for other observables, notably the number variance; see \cite{ADM24}.
\end{rem}

See Figure~\ref{Fig_DSFF mesocopic} for numerical illustrations, particularly in the ramp regime, where linear, quadratic, and mixed behaviours are observed.

\begin{figure}[t]
\centering

\begin{subfigure}{0.32\textwidth}
    \centering
    \includegraphics[width=\linewidth]{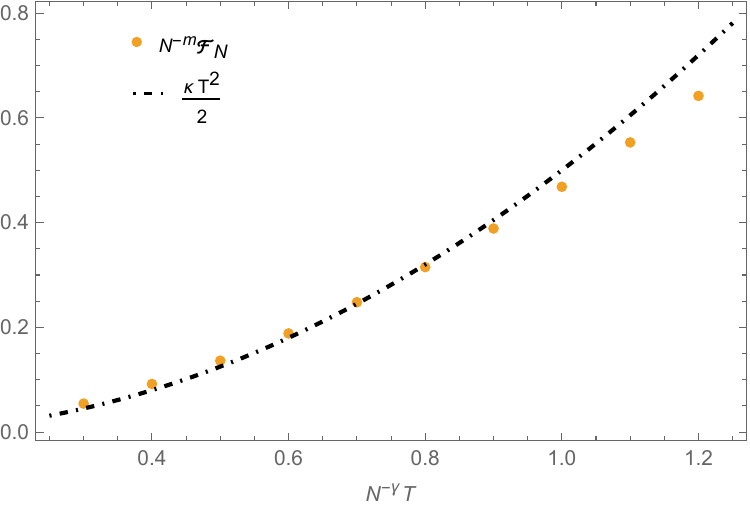}
    \caption{$\alpha=0.3$, $\kappa=1$}
\end{subfigure}
\hfill
\begin{subfigure}{0.32\textwidth}
    \centering
    \includegraphics[width=\linewidth]{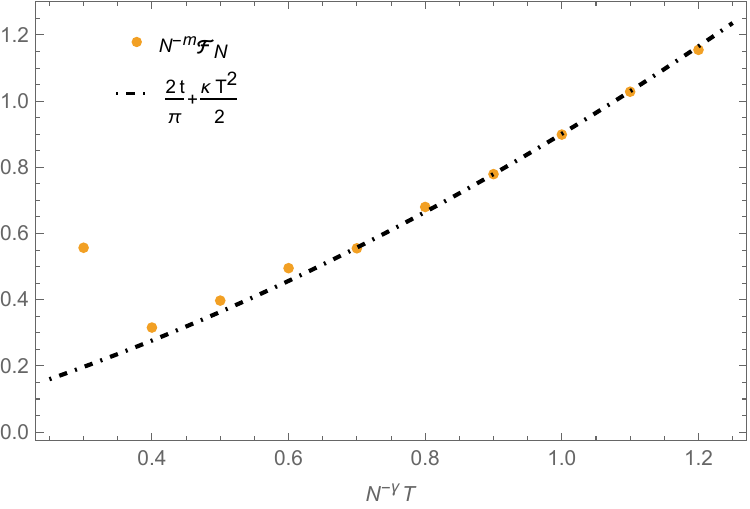}
    \caption{$\alpha=0.6$, $\kappa=0.7$}
\end{subfigure}
\hfill
\begin{subfigure}{0.32\textwidth}
    \centering
    \includegraphics[width=\linewidth]{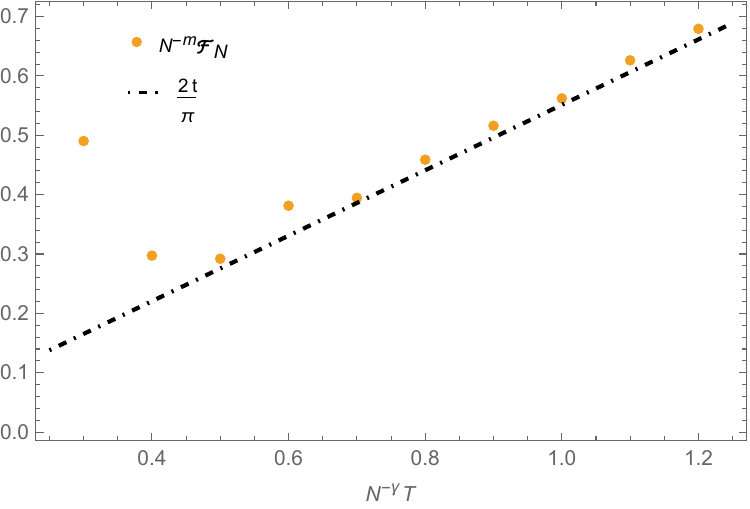}
    \caption{$\alpha=0.9$, $\kappa=0.4$}
\end{subfigure}

\caption{The plots illustrate the DSFF in the mesoscopic non-Hermiticity regime with parameters $N=2^{14}$, $\theta=\pi/6$, and $\gamma=0.6$,
in all 3 figures.
Accordingly, in case (A) with $\alpha < \gamma$, one observes the quadratic ramp $\frac{\kappa \mathsf{T}^2}{2}$, whereas in case (C) with $\alpha > \gamma$, the behaviour is given by the linear ramp $\frac{2\mathsf{t}}{\pi}$. In the critical case (B), where $\alpha = \gamma$, the resulting behaviour is a combination of these two regimes; see \eqref{mesoscopic nH conn} and Figure~\ref{fig:alpha_gamma} (B). 
The $y$-axis is rescaled by $N^m$, where $m = \max\{2\gamma - \alpha, \gamma\}$. The orange dots represent the DSFF evaluated at finite $N$, while the black dot-dashed curves indicate the corresponding analytic limits.} 
\label{Fig_DSFF mesocopic}
\end{figure}

\medskip

As our final main result, we analyse the asymptotic behaviour of the DSFF in the regime of weak non-Hermiticity. As before, we add subscript ``w'' denoting the weak non-Hermiticity.

\begin{thm}[\textbf{Asymptotics of the disconnected/connected DSFF at weak non-Hermitcity}] \label{thm weak nH}
Let $\tau$ be scaled as \eqref{def of tau scaling general} with $\alpha \ge 1.$

\begin{itemize}
    \item[(i)] \textup{\textbf{(Dip-ramp regime)}} Let $\gamma\in[0,1]$. Then as $N\to\infty$, we have
\begin{equation} \label{asymp of FN d/c wnH}
    \mathcal{F}_{N}^{(\textup{d})}(N^{\gamma}\mathsf{T},\theta) = N^{2-3\gamma}    \Big(  \mathcal{F}_{ \textup{w}}^{(\textup{d})}(\mathsf{T},\theta) + O(N^{-\epsilon_5})\Big),
    \qquad \mathcal{F}_{N}^{(\textup{c})}(N^{\gamma}\mathsf{T},\theta) = N^{\gamma} \Big(\mathcal{F}_{\textup{w}}^{(\textup{c})}(\mathsf{T},\theta)  + O(N^{-\epsilon_6})\Big),
\end{equation}
where  
\begin{equation} \label{weak nH dis}
    \mathcal{F}_{\textup{w}}^{(\textup{d})}(\mathsf{T},\theta)  := \begin{cases} \displaystyle
        \frac{J_{1} (2\mathsf{t})^2}{\mathsf{t}^2} & \textup{for } \gamma=0,
        \smallskip 
        \\
        \displaystyle
        \frac{1}{2\pi \mathsf{t}^3} \Big(1-\sin(4N^{\gamma}\mathsf{t})\Big) & \textup{for } \gamma\in(0,1), 
        \smallskip
        \\
        \displaystyle
        \exp\Big(-N^{1-\alpha}\frac{\kappa}{2}\mathsf{T}^2\Big) \frac{1}{2\pi \mathsf{t}^3} \Big(1-\sin\Big[N\Big(\mathsf{t}\sqrt{4-\mathsf{t}^2}+4\,\textup{arcsin}\frac{\mathsf{t}}{2}\Big)+O(N^{1-\alpha})\Big]\Big) & \textup{for } \gamma=1, \, \mathsf{t}<2
    \end{cases}
\end{equation}
and
\begin{equation} \label{weak nH conn}
    \mathcal{F}_{\textup{w}}^{(\textup{c})}(\mathsf{T},\theta)  := \begin{cases} \displaystyle
        \mathsf{t}^2 \Big( 2J_0(2\mathsf{t})^2 + 2 J_1(2\mathsf{t})^2 - \frac{J_0(2\mathsf{t})J_1(2\mathsf{t})}{\mathsf{t}} \Big) & \textup{for } \gamma=0,
        \smallskip 
        \\
        \displaystyle
        \frac{2}{\pi}\mathsf{t} & \textup{for } \gamma\in(0,1),
        \smallskip
        \\
        \displaystyle
        1-\exp\Big(-N^{1-\alpha}\frac{\kappa}{2}\mathsf{T}^2\Big)+\frac{1}{2\pi}\Big(\mathsf{t}\sqrt{4-\mathsf{t}^2}+4\,\textup{arcsin}\frac{\mathsf{t}}{2}\Big) & \textup{for } \gamma=1, \, \mathsf{t}<2.
    \end{cases}
\end{equation}
In \eqref{asymp of FN d/c wnH}, the constants $\epsilon_5 , \epsilon_6  $ are given by 
\begin{equation} \label{def of epsilon 5,6}
    \epsilon_5 = \min\{2-3\gamma,2\gamma,\alpha-\gamma\}, \qquad
    \epsilon_6 = \begin{cases} \displaystyle
        \min\{2,\alpha\} & \textup{for } \gamma=0,
        \smallskip 
        \\ \displaystyle
        \min\{2\gamma,2-2\gamma,\alpha-\gamma\} & \textup{for } \gamma\in(0,1),
        \smallskip 
        \\ \displaystyle
        1 & \textup{for } \gamma=1.  
    \end{cases}
\end{equation} %\smallskip 
\item[(ii)] \textup{\textbf{(Plateau regime)}} Let $\gamma\ge1$. Then as $N \to \infty$, we have
\begin{equation} \label{asymp of FN p wnH}
\begin{split}
    & \mathcal{F}_{N}^{(\textup{d})}(N^{\gamma}\mathsf{T},\theta) = \exp\Big(-N^{2\gamma-1} \Phi_{\textup{w}}(\mathsf{t})+O(\mathsf{e}_3(N))\Big), \\
    & \mathcal{F}_{N}^{(\textup{c})}(N^{\gamma}\mathsf{T},\theta) = N-\exp\Big(-N^{2\gamma-1} \Phi_{\textup{w}}(\mathsf{t})+O(\mathsf{e}_3(N))\Big),
\end{split}
\end{equation}
where
\begin{equation}
\Phi_{\textup{w}}(\mathsf{t}) := \begin{cases} \displaystyle 
    \mathsf{t} \sqrt{\mathsf{t}^2-4}-4\,\textup{arccosh}\frac{\mathsf{t}}{2} & \textup{for } \gamma=1, \, \mathsf{t}>2, \smallskip \\ \displaystyle 
    \mathsf{t}^2 & \textup{for } \gamma>1.
    \end{cases}
\end{equation}
Here, the error term $\mathsf{e}_3(N)$ is given by 
\begin{equation} \label{def of e3}
    \mathsf{e}_3(N) = \begin{cases} \displaystyle
        \log N &\textup{for } \gamma=1, \, \mathsf{t}>2, \smallskip \\ \displaystyle
        N\log N &\textup{for } \gamma\in(1,\frac{2+\alpha}{2}), \smallskip \\ \displaystyle
        N^{2\gamma-1-\alpha} &\textup{for } \gamma>\frac{2+\alpha}{2}.
    \end{cases}
\end{equation}
\end{itemize}
\end{thm}

We mention that in Theorems~\ref{thm strong nH}, \ref{thm mesoscopic nH}, and \ref{thm weak nH}, the stated error bounds are indeed optimal.  

Note that, as before, the resulting formulas are expressed in terms of $\mathsf{T}$ and $\mathsf{t}$. Moreover, for $\alpha > 1$, the dependence on $\mathsf{T}$ is suppressed, so that the formulas depend only on $\mathsf{t}$. In this regime, the asymptotic behaviour of the DSFF is universal, coinciding with that of the GUE. This reflects the fact that the model is already sufficiently close to the Hermitian setting at the level of DSFF asymptotics.

More precisely, the asymptotics in Theorem~\ref{thm weak nH} are consistent with the known results for the spectral form factor of the GUE~\cite{Fo21}. In particular, the regimes $\gamma = 0$ and $\gamma = 1$ correspond to \cite[Propositions~16 and~18]{Fo21}, respectively. Furthermore, the linear ramp behaviour observed in~\eqref{weak nH conn} for $0 < \gamma < 1$ agrees with \cite[Remark~17]{Fo21}, where it arises by taking the limit $t \to \infty$ in the $\gamma = 0$ case. 

As an immediate consequence of Theorem~\ref{thm weak nH}, we have the following. 

\begin{cor}[\textbf{Leading order asymptotics of the DSFF at weak non-Hermitcity}]  \label{Cor_DSFF weak} Let $\tau$ be scaled as \eqref{def of tau scaling general} with $\alpha \ge 1.$
As $N \to \infty$, we have 
\begin{equation}
\mathcal{F}_N(N^{\gamma}\mathsf{T},\theta) \sim 
\begin{cases}
N^{2-3\gamma} \mathcal{F}_{ \rm w }^{\rm (d)} (\mathsf{T},\theta) &\textup{for } \gamma \in [0, \frac12 ), 
\smallskip 
\\
N^{\frac12 } \big( \mathcal{F}_{ \rm w }^{\rm (d)} (\mathsf{T},\theta)+  \mathcal{F}_{ \rm w }^{\rm (c)} (\mathsf{T},\theta) \big)  &\textup{for } \gamma = \frac12,
\smallskip 
\\
N^{ \gamma }  \mathcal{F}_{ \rm w }^{\rm (c)} (\mathsf{T},\theta)  &\textup{for } \gamma \in (\frac12, 1), 
\smallskip 
\\
N \, f(\mathsf{t})
\smallskip &\textup{for } \gamma =1, 
\\
N &\textup{for } \gamma >1. 
\end{cases}
\end{equation}
Here, 
\begin{equation}
f(\mathsf{t}) :=    \frac{2}{\pi}\Big(\frac{\mathsf{t}\sqrt{4-\mathsf{t}^2}}{4}+\textup{arcsin}\frac{\mathsf{t}}{2}\Big) \mathbbm{1}_{ \{ |\mathsf{t}|<2 \} } + \mathbbm{1}_{ \{ \mathsf{t}>2 \} }. 
\end{equation}
\end{cor}

\begin{rem}[Thouless and Heisenberg time] 
Combining Corollaries~\ref{Cor_DSFF strong}, \ref{Cor_DSFF mesoscopic}, and~\ref{Cor_DSFF weak}, we observe that, under the general scaling assumption~\eqref{def of tau scaling general}, the Thouless and Heisenberg times obey the asymptotic relations
\begin{equation}
    T_{\text{Th}} \propto N^{ \min\{ \frac{2+\alpha}{5} , \frac12 \}  } ,     \qquad T_{\text{H}} \propto N^{ \min\{ \frac{1+\alpha}{2} , 1 \}  } .
\end{equation}
In particular, the Heisenberg time corresponds to the reciprocal of the
typical eigenvalue spacing given in~\eqref{def of reciprocal spacing}. 
\end{rem}

\subsection*{Organisation of the paper} The remainder of this paper is organised as follows. In Section~\ref{Section_finite}, we present several finite-$N$ formulas, most of which follow from algebraic computations, and which yield a representation of the DSFF amenable to asymptotic analysis. In Section~\ref{Section_asymptotic}, building on these finite-$N$ formulas, we carry out a precise asymptotic analysis based on orthogonal polynomial techniques.

\subsection*{Acknowledgments} The work of Gernot Akemann was partly funded by the Deutsche Forschungsgemeinschaft (DFG, German Research Foundation) – Project-ID 317210226 – SFB 1283. Sung-Soo Byun was supported by the National Research Foundation of Korea grants (RS-2025-00516909).

\section{Finite-$N$ analysis} \label{Section_finite}

In this section, we prepare several preliminaries for our main results.  
For this purpose, we first consider a general point process 
$\boldsymbol{z}=\{z_j\}_{j=1}^N$ in the complex plane, defined through the joint probability distribution
\begin{equation} \label{def of Gibbs}
    P_N(\boldsymbol{z}) \prod_{j=1}^N\,dA(z_j) = \frac{1}{Z_N} \prod_{j<k} |z_j-z_k|^2 \prod_{j=1}^N \omega(z_j)\,dA(z_j). 
\end{equation}
Here, $\omega(z)$ denotes a weight function specifying the ensemble and $Z_N$ is the normalisation constant.
In particular, the eigenvalues of the complex elliptic Ginibre 
%matrix 
ensemble follow the distribution \eqref{def of Gibbs} with the weight function
\begin{equation} \label{def of elliptic weight}
 \omega^{\rm H}_N(z) : = \exp\Big(-N \frac{ |z|^2-\tau \re z^2 }{1-\tau^2}  \Big). 
\end{equation}
The $k$-point correlation function of the ensemble~\eqref{def of Gibbs}, which can be characterised via~\eqref{def of RNk test}, can also be defined in terms of the joint probability distribution by integrating out $N-k$ variables: 
\begin{equation} \label{def of det structure}
R_{N,k}(z_1,\dots,z_k) = \frac{N!}{(N-k)!} \int_{ \C^{N-k} } P_N(\boldsymbol{z}) \prod_{j=k+1}^N \,dA(z_j). 
\end{equation}

A remarkable property of the ensemble \eqref{def of Gibbs} is that it forms a determinantal point process. Namely, we have the representation \eqref{def of RNk det KN} 
where 
\begin{equation} \label{def of KN gen OP}
    K_N(z,z') := \sqrt{\omega(z)\omega(z')} \sum_{j=0}^{N-1} \phi_j(z) \overline{\phi_j(z')}.
\end{equation}
Here, $\{\phi_n\}_{n\in\mathbb{N}}$ is the set of orthonormal polynomials associated to the weight function $\omega(z)$, i.e. 
\begin{equation} \label{def of orthonormality}
\int_\C \phi_n(z) \overline{\phi_m(z)} \omega(z)\,dA(z) = \delta_{n,m}. 
\end{equation}
In particular, for the elliptic weight \eqref{def of elliptic weight}, the associated planar orthonormal polynomial is given by (see e.g. \cite{EM90,DGIL94} and \cite[Lemma 7]{ACV18}) 
\begin{equation}
\phi^{\rm H}_{N,n}(z) = \frac{1}{\sqrt{1-\tau^2}\sqrt{n!}} \Big(\frac{\tau}{2}\Big)^{n/2} H_n\Big(\sqrt{\frac{N}{2\tau}}z\Big),
\end{equation} 
where 
\begin{equation}
H_n(x)  := (-1)^n e^{x^2} \frac{d^n}{dx^n} e^{-x^2} 
\end{equation}
is the Hermite polynomial. 

For the general ensemble defined in \eqref{def of Gibbs}, the associated DSFF is introduced through \eqref{def of DSFF}. By combining \eqref{def of RNk det KN}, \eqref{def of FN disconn in terms of RN1 int}, \eqref{def of FN conn in terms of RN2 int}, and \eqref{def of KN gen OP}, one can immediately derive exact expressions for both the disconnected and connected parts of the DSFF in terms of the underlying orthonormal polynomials.

\begin{lem} \label{Lem_finite N general}
 We have 
    \begin{equation} \label{dis, conn def}
        \mathcal{F}_N^{\rm (d)} = \bigg\vert \sum_{j=0}^{N-1} F_{jj} \bigg\vert^2, \qquad \mathcal{F}_N^{\rm (c)} = N - \sum_{j,k=0}^{N-1} \big\vert F_{jk} \big\vert^2,
    \end{equation}
    where 
    \begin{equation} \label{Fjk def}
        F_{jk} := \int_{ \C } e^{itx+isy} \phi_j(z) \overline{\phi_k(z)} \omega(z) \, dA(z).
    \end{equation}
\end{lem}

\begin{comment} 
\begin{proof}
Due to the determinantal structure, the $1$- and $2$-point function is given by 
\begin{align}
R_{N,1}(z) &= K_N(z,z) = \omega(z) \sum_{ j=0 }^{N-1} |\phi_j(z)|^2,
\\
R_{N,2}(z,w) &= R_{N,1}(z)R_{N,1}(w) - K_N(z,w)K_N(w,z). 
\end{align}
Note that by definition \eqref{def of DSFF disc, conn}, the averaged disconnected part is givn in terms of the $1$-point function as 
\begin{align*}
\langle Z_N \rangle = \bigg\langle \sum_{j=1}^N e^{i t x_j+is y_j} \bigg\rangle  = \sum_{j=1}^N \int_{ \C } e^{it x_j + is y_j} R_{N,1}(z)\,dA(z) = \sum_{j=0}^{N-1} F_{jj}. 
\end{align*}
This gives the desired formula for $\mathcal{F}_N^{ \rm (d) }.$

On the other hand, note that 
\begin{align*}
  \mathcal{F}_N  & = N + \sum_{\substack{j,k=1 \\ j\ne k}}^N \int_{\mathbb{C}^N} e^{it(x_j-x_k)+is(y_j-y_k)} R_{N}(z_1,\cdots,z_N) \, dA(z_1)\cdots dA(z_N) \\
    & = N + \int_{ \C^2 } e^{it(x_1-x_2)+is(y_1-y_2)} R_{N,2}(z_1,z_2) \, dA(z_1) \, dA(z_2).
\end{align*}
Therefore it follows that  
\begin{align*}
\mathcal{F}_N^{ \rm (c) } = N - \int_{ \C^2 }  e^{it(x_1-x_2)+is(y_1-y_2)} |K_N(z_1,z_2)|^2 \,dA(z_1)\,dA(z_2) =N - \sum_{j,k=0}^{N-1} \big\vert F_{jk} \big\vert^2, 
\end{align*}
which completes the proof.  
\end{proof} 
\end{comment}

Next, we evaluate \(F_{jk}\) in \eqref{Fjk def} for the elliptic Ginibre ensemble, which we denote by
\begin{equation}
    F^{\rm H}_{jk} := \int_{\mathbb{C}} e^{itx+isy} \phi^{\rm H}_{N,j}(z) \overline{\phi^{\rm H}_{N,k}(z)} \omega^{\rm H}_N(z) \, dA(z).
\end{equation}
This quantity can be evaluated explicitly in terms of the generalised Laguerre polynomials, defined by
\begin{equation} \label{def of Laguerre}
L^{(\nu)}_k(x) := \frac{x^{-\nu}e^x}{k!} \frac{d^k}{dx^k} (x^{k+\nu} e^{-x})=\sum_{k=0}^{j}\frac{\Gamma(j+\nu+1)}{(j-k)!\Gamma(\nu+k+1)}\frac{(-z)^k}{k!}.  
\end{equation}
Note that it satisfies the orthogonality relation 
\begin{equation} \label{def of orthogonality of Laguerre}
\int_0^\infty x^\nu e^{-x} L_n^{(\nu)}(x) L_m^{(\nu)}(x)\,dx= \frac{\Gamma(n+\nu+1)}{ n! } \delta_{n,m},
\end{equation}
where $\delta$ is the Kronecker delta. 

\begin{prop}\label{Prop_eval of Fjk Hermite Laguerre}
Let $\tau \in [0,1)$. Then we have  
    \begin{equation} \label{Fjk elliptic}
        F^{\rm H}_{jk} = \exp\Big(-\frac{1-\tau^2}{8N}T^2-\frac{|\eta T|^2}{2N}\Big) \sqrt{\frac{k!}{j!}} \Big(\frac{i\eta T}{\sqrt{N}}\Big)^{j-k} L^{(j-k)}_{k}\Big(\frac{\vert\eta T\vert^2}{N}\Big), 
    \end{equation}
    where $T$ and $\eta$ are defined in \eqref{def of complex time variables} and \eqref{def of eta}, respectively. 
\end{prop}

We note that Proposition~\ref{Prop_eval of Fjk Hermite Laguerre} remains valid in the Hermitian limit \(\tau=1\).
In this case, the identity is well known; see e.g. \cite[Proposition~13]{Fo21} and the references therein.
By contrast, for the planar case \(\tau \in [0,1)\), this identity does not appear to be explicitly recorded in the literature, although it was used implicitly in the supplementary material of \cite{SKSK24}. %, where partial reliance on computer algebra is mentioned.

\begin{proof}[Proof of Proposition~\ref{Prop_eval of Fjk Hermite Laguerre}]
The complex exponential \(e^{itx+isy}\) can be absorbed into the weight function, which leads to a convenient translation of variables. Indeed, we have 
\begin{equation*}
\begin{split}
    F^{\rm H}_{jk} & = \int_{\mathbb{C}} \phi^{\rm H}_{N,j}(z)\phi^{\rm H}_{N,k}(\overline{z}) \exp\Big(-\frac{Nx^2}{1+\tau}-\frac{Ny^2}{1-\tau}+itx+isy\Big) \, dA(z) \\
    & = \exp\Big(-\frac{1+\tau}{4N}t^2-\frac{1-\tau}{4N}s^2\Big) \int_{\mathbb{C}} \phi^{\rm H}_{N,j}(z)\phi^{\rm H}_{N,k}(\overline{z}) \exp\Big(-\frac{N(x-\frac{it(1+\tau)}{2N})^2}{1+\tau}-\frac{N(y-\frac{is(1-\tau)}{2N})^2}{1-\tau}\Big) \, dA(z) \\
    & = \exp\Big(-\frac{1-\tau^2}{8N}T^2-\frac{|\eta T|^2}{2N}\Big) \int_{\mathbb{C}} \phi^{\rm H}_{N,j}\Big(z+\frac{i\eta T}{N}\Big)\phi^{\rm H}_{N,k}\Big(\overline{z}+\frac{i\overline{\eta}T}{N}\Big) \exp\Big(-\frac{Nx^2}{1+\tau}-\frac{Ny^2}{1-\tau}\Big) \, dA(z).
\end{split}
\end{equation*}
Using the identity
\begin{equation*}
    H_n(a+b) = \sum_{l=0}^n \binom{n}{l} H_l(a) (2b)^{n-l}
\end{equation*}
together with the orthonormality \eqref{def of orthonormality} of the family \((\phi^{\rm H}_{N,n})_{n\in\mathbb{N}}\), we obtain
\begin{equation*}
    F^{\rm H}_{jk} = \exp\Big(-\frac{1+\tau}{4N}t^2-\frac{1-\tau}{4N}s^2\Big) \frac{1}{\sqrt{j!\,k!}} \sum_{l=0}^{j\wedge k} l!\binom{j}{l}\binom{k}{l} \Big(\frac{i\eta T}{\sqrt{N}}\Big)^{j-l} \Big(\frac{i\overline{\eta}T}{\sqrt{N}}\Big)^{k-l}.
\end{equation*}
If \(j\ge k\), this expression simplifies to
\begin{equation*}
\begin{split}
    F^{\rm H}_{jk} & = \exp\Big(-\frac{1+\tau}{4N}t^2-\frac{1-\tau}{4N}s^2\Big) \sqrt{\frac{k!}{j!}} \Big(\frac{i\eta T}{\sqrt{N}}\Big)^{j-k} \sum_{l=0}^{k} \binom{j}{l}\frac{1}{(k-l)!} \Big(-\frac{|\eta T|^2}{N}\Big)^{k-l} \\
    & = \exp\Big(-\frac{1+\tau}{4N}t^2-\frac{1-\tau}{4N}s^2\Big) \sqrt{\frac{k!}{j!}} \Big(\frac{i\eta T}{\sqrt{N}}\Big)^{j-k} L^{(j-k)}_{k}\Big(\frac{\vert\eta T\vert^2}{N}\Big).
\end{split}
\end{equation*}
The case \(j<k\) follows by symmetry using the identity
\begin{equation*}
    \frac{(-x)^m}{m!}L^{(m-n)}_n(x) = \frac{(-x)^n}{n!}L^{(n-m)}_m(x).
\end{equation*}
This completes the proof. 
\end{proof}

As an immediate consequence of Lemma~\ref{Lem_finite N general} and
Proposition~\ref{Prop_eval of Fjk Hermite Laguerre}, together with the
summation identity for Laguerre polynomials 
\begin{equation}
    L^{(\nu+1)}_n(x) = \sum_{k=0}^{n} L^{(\nu)}_k(x), 
\end{equation}
we obtain the following, which also appears in \cite[Eqs.~(5), (6)]{SKSK24}.

\begin{prop} \label{prop finite N}
For all positive integers $N$ and $\tau \in [0,1)$, the following exact identities hold. 
\begin{itemize}
    \item We have  
    \begin{equation} \label{dis finite}
        \mathcal{F}_{N}^{(\textup{d})} = \exp\Big(-\frac{1-\tau^2}{4N}T^2\Big) f_N\Big(\frac{\vert\eta T\vert^2}{N}\Big),
    \end{equation}
    where 
    \begin{equation} \label{def of fN}
     f_N(x) := e^{-x} L_{N-1}^{(1)}(x)^2. 
    \end{equation}
    \item We have
    \begin{equation} \label{conn finite}
        \mathcal{F}_{N}^{(\textup{c})} = N - \exp\Big(-\frac{1-\tau^2}{4N}T^2\Big) \Psi_N\Big(\frac{\vert\eta T\vert^2}{N}\Big),
    \end{equation}
    where 
\begin{equation} \label{def of PsiN}
 \Psi_N(x) := e^{-x}  \sum_{j,k=0}^{N-1} (-1)^{j-k} L_j^{(k-j)}(x) L_k^{(j-k)}(x) . 
\end{equation}
\end{itemize}
\end{prop}

\begin{comment} 
\begin{proof}
Given \eqref{Fjk elliptic}, the disconnected part becomes
\begin{equation*}
    \mathcal{F}^{\rm H}_{\text{dis}} = \exp\Big(-\frac{1-\tau^2}{4N}T^2-\frac{|\eta T|^2}{N}\Big) \bigg\vert \sum_{j=0}^{N-1} L^{(0)}_{j}\Big(\frac{\vert\eta T\vert^2}{N}\Big) \bigg\vert^2
\end{equation*}
Then, \eqref{dis finite} follows immediately from the summation formula of the Laguerre polynomials:
\begin{equation*}
    L^{(\nu+1)}_n(x) = \sum_{k=0}^{n} L^{(\nu)}_k(x).
\end{equation*}

Next, observe that
\begin{equation*}
    \overline{F^{\rm H}_{jk}} = (-1)^{j-k} F^{\rm H}_{kj} = (-1)^{j-k} \exp\Big(-\frac{1-\tau^2}{8N}T^2-\frac{|\eta T|^2}{2N}\Big) \sqrt{\frac{j!}{k!}} \Big(\frac{i\eta T}{\sqrt{N}}\Big)^{k-j} L^{(k-j)}_{j}\Big(\frac{\vert\eta T\vert^2}{N}\Big),
\end{equation*}
which follows from the definition of $F^{\rm H}_{jk}$ and the property of the Hermite polynomials $H_n(-x)=(-1)^nH_n(x)$. Thus we have
\begin{equation*}
    \vert F^{\rm H}_{jk}\vert^2 = \exp\Big(-\frac{1-\tau^2}{4N}T^2-\frac{\vert\eta T\vert^2}{N}\Big) (-1)^{j-k} L^{k-j}_j\Big(\frac{\vert\eta T\vert^2}{N}\Big) L^{j-k}_k\Big(\frac{\vert\eta T\vert^2}{N}\Big).
\end{equation*} 
\end{proof}

\end{comment}

We discuss an important structural relation underlying the functions $f_N$ and $\Psi_N$ defined in \eqref{def of fN} and \eqref{def of PsiN}. For this purpose, we denote by 
\begin{equation}  \label{def of rho}
    \rho_N(x) = e^{-x} \sum_{k=0}^{N-1} L_k^{(0)}(x)^2
\end{equation}
the one-point function of the 
Laguerre unitary ensemble (LUE) 
with rectangular parameter equal to zero.
By the Christoffel--Darboux formula, \(\rho_N\) can be rewritten as
\begin{equation}
\rho_N(x) = N e^{-x} \Big( L_{N-1}^{(0)}(x)L_{N-1}^{(1)}(x) - L_N^{(0)}(x) L^{(1)}_{N-2}(x)\Big), 
\end{equation}
where we have used the differentiation identity
\begin{equation} \label{derivative of Laguerre}
\frac{d}{dx} L_N^{(\nu)}(x) = -L_{N-1}^{(\nu+1)}(x). 
\end{equation}
A remarkable feature is that $f_N$ and $\Psi_N$ satisfy the following simple differentiation relations, in which $\rho_N$ plays a central intermediary role. 
This identity goes back to Brézin and Hikami~\cite[Eq.~(4.6)]{BH97}; see also  \cite[Eq.~(2.19)]{Oku19} and \cite[Eq.~(3.20)]{Fo21}.

\begin{prop} \label{Prop_BH relation}
    We have
    \begin{align}
     \Psi_N(x) & = \int_x^\infty \rho_N(u) \,du = N - \int_{0}^x \rho_N(u) \,du,
        \\
    \rho_N(x) &= \int_x^\infty f_N(u) \,du = N - \int_{0}^x f_N(u) \,du.
    \end{align} 
\end{prop} 

In Proposition~\ref{Prop_Psi alternative}, we present alternative representations of the function $\Psi_N$. Although these alternative representations are not used in the asymptotic analysis in the subsequent section, they involve certain functional identities for Laguerre functions that may be of independent interest and potentially useful in other contexts. 

Indeed, Proposition~\ref{Prop_BH relation} was used in~\cite{Fo21} to analyse the asymptotic behaviour of the SFF for the GUE. More precisely, one begins with the asymptotics of $\rho_N$, obtained from the corresponding asymptotics of the density of the  %Laguerre unitary ensemble (LUE)
LUE, and then integrates once to derive the asymptotic behaviour of $\Psi_N$; see Remark~\ref{Rem_LUE density} for further details. 

We also note that the function $f_N$ can be identified with the generating function of spectral moments of the LUE; see \cite[Eq.~(4.28)]{CMOS19}. 

%However, at the level of precision required in the present setting, this approach leads to substantial technical difficulties. In particular, performing two successive integrations of asymptotic expansions makes it challenging to retain uniform control over the associated error terms. This difficulty is compounded by the fact that Laguerre polynomials exhibit qualitatively different asymptotic behaviours across different regimes of the argument. As a result, one is led to employ uniform asymptotic expansions of the Laguerre polynomials with sharply controlled error bounds, which will be developed in detail in the next section.

\medskip 

\section{Large-$N$ asymptotic analysis}   \label{Section_asymptotic}

In this section, we perform the asymptotic analysis and complete the proofs of the main theorems. 
By Proposition~\ref{prop finite N}, it suffices to establish the precise asymptotic behaviour of the quantities $f_N$ and $\Psi_N$. 

In Subsection~\ref{Subsec_asymptotic of fN PsiN proofs of main}, we first state the asymptotic expansions of $f_N$ and $\Psi_N$ in Propositions~\ref{prop asymp f} and~\ref{prop asymp Psi}, respectively. We then combine these results to prove our main theorems, Theorems~\ref{thm strong nH}, \ref{thm mesoscopic nH}, and~\ref{thm weak nH}. 
The proofs of Propositions~\ref{prop asymp f} and~\ref{prop asymp Psi} are deferred to the subsequent subsections, Subsection~\ref{Subsec_proof of Prop asymp f_N} and Subsection~\ref{Subsectoin_prop asymp rho Psi}.

\subsection{Proof of main theorems} \label{Subsec_asymptotic of fN PsiN proofs of main}

By the exact identities derived in the previous section, our asymptotic analysis reduces to obtaining uniform asymptotics for the Laguerre polynomials. To this end, it is necessary to distinguish several asymptotic regions. More precisely, for fixed constants $c,d>0$, we define
\begin{equation} \label{def of regions I II III IV}
    \RN{1} := [0,c], \qquad \RN{2} := [c,4N-d\sqrt{N}], \qquad \RN{3} := [4N-d\sqrt{N},4N+d\sqrt{N}], \qquad \RN{4} := [4N+d\sqrt{N},\infty) 
\end{equation}
which we refer to as the \emph{Bessel regime}, \emph{oscillatory regime}, \emph{Airy regime}, and \emph{exponential regime}, respectively. See Figure~\ref{Fig_Laguerre regimes}. 

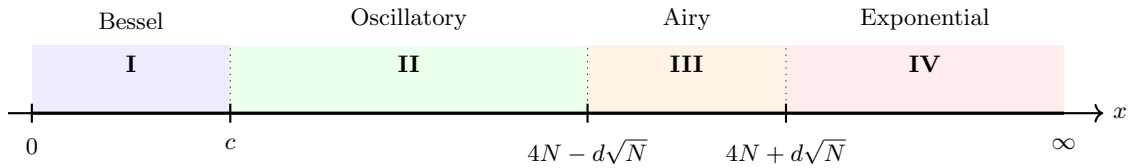
\begin{figure}[h!]
\centering
\begin{tikzpicture}[scale=1.05, every node/.style={font=\small}]
    % key points
    \coordinate (A) at (0,0);
    \coordinate (B) at (2.5,0);
    \coordinate (C) at (7.0,0);
    \coordinate (D) at (9.5,0);
    \coordinate (E) at (13.0,0);

    % shaded regions
    \fill[blue!7]    (A |- 0,0) rectangle (B |- 0,0.85);
    \fill[green!7]   (B |- 0,0) rectangle (C |- 0,0.85);
    \fill[orange!10] (C |- 0,0) rectangle (D |- 0,0.85);
    \fill[red!7]     (D |- 0,0) rectangle (E |- 0,0.85);

    % axis
    \draw[->, thick] (-0.3,0) -- (13.5,0) node[right] {$x$};

    % interval line
    \draw[line width=1.2pt] (A) -- (E);

    % ticks
    \foreach \P/\lab in {
        A/$0$,
        B/$c$,
        C/$4N-d\sqrt N$,
        D/$4N+d\sqrt N$
    }{
        \draw[thick] (\P) ++(0,-0.12) -- ++(0,0.24);
        \node[below=6pt] at (\P) {\lab};
    }

    % boundary guides
    \foreach \P in {B,C,D}{
        \draw[dotted] (\P |- 0,0) -- (\P |- 0,0.85);
    }

    % interval labels
    \node[above=12pt] at ($(A)!0.5!(B)$) {\textbf{I}};
    \node[above=28pt] at ($(A)!0.5!(B)$) {Bessel};

    \node[above=12pt] at ($(B)!0.5!(C)$) {\textbf{II}};
    \node[above=28pt] at ($(B)!0.5!(C)$) {Oscillatory};

    \node[above=12pt] at ($(C)!0.5!(D)$) {\textbf{III}};
    \node[above=28pt] at ($(C)!0.5!(D)$) {Airy};

    \node[above=12pt] at ($(D)!0.5!(E)$) {\textbf{IV}};
    \node[above=28pt] at ($(D)!0.5!(E)$) {Exponential};

    % infinity
    \node[below=6pt] at (E) {$\infty$};

\end{tikzpicture}
\caption{Decomposition of the positive real axis into the four asymptotic regimes for the Laguerre polynomials.}
\label{Fig_Laguerre regimes}
\end{figure} 

Recall that $f_N(x)$ is defined in \eqref{def of fN}. We first derive its uniform asymptotic behaviour in all relevant regimes, with sufficient precision for the subsequent analysis. In particular, these estimates will allow us to obtain the asymptotic behaviours of $\Psi_N$ and $\rho_N$ via Proposition~\ref{Prop_BH relation}. For this purpose, we recall that the Airy function is defined by
\begin{equation}
    \Ai(x) = \frac{1}{\pi}\int_0^\infty\cos\Big(\frac{t^3}{3}+xt\Big)\,dt;
\end{equation}
see e.g. \cite[Chapter 9]{NIST}. Moreover, we introduce the following notations. First, we define 
 \begin{equation} \label{def xi}
        \xi(x):=\frac{1}{2}\big(\sqrt{x-x^2}+\arcsin\sqrt{x}\big).
    \end{equation} 
For $p>0$, we also write 
   \begin{equation} \label{def of envAi}
        \textup{envAi}_p(x) = \begin{cases} \displaystyle
            x^p \exp\Big(-\frac{2}{3}x^{3/2}+\frac{2}{3}\Big) & \text{if }x\ge1, 
            \smallskip 
            \\ \displaystyle
            1 & \text{if }-1\le x<1,
            \smallskip 
            \\ \displaystyle
            (-x)^p & \text{if } x<-1,
        \end{cases}
    \end{equation}
where $\operatorname{env}$ denotes the envelope.

\medskip 

\begin{prop}[\textbf{Asymptotics of }$f_N$] \label{prop asymp f}
Let $\mathsf{x}=x/4N$, $\mathsf{X}=4Nx$ and $y=(2N)^{2/3}(\mathsf{x}-1)$.
\begin{itemize}
    \item[\textup{(i)}] \textup{\textbf{(Bessel regime)}} As $N\to\infty$, we have   
    \begin{align}
    \begin{split} \label{asymp f bessel}
         f_N(x) & =  4N^2\frac{J_1(\sqrt{\mathsf{X}})^2}{\mathsf{X}} + \frac{1}{12}\sqrt{\mathsf{X}}J_1(\sqrt{\mathsf{X}})J_2(\sqrt{\mathsf{X}}) 
        \\
        &\quad  + \frac{1}{11520\,N^2}
\Big[\mathsf X
(24-5\mathsf X)J_1(\sqrt{\mathsf X})^2
-24\sqrt{\mathsf X}(4-\mathsf X)
J_1(\sqrt{\mathsf X})J_2(\sqrt{\mathsf X})
+5\mathsf X^2 J_2(\sqrt{\mathsf X})^2
\Big] + O(N^{-1})
    \end{split}
    \end{align} 
    uniformly for $x\in\RN{1}$. 
    \smallskip 
    \item[\textup{(ii)}] \textup{\textbf{(Oscillatory regime)}} As $N\to\infty$, we have 
    \begin{equation} \label{asymp f cos}
    \begin{split}
        f_N(x)&  =  \frac{1}{4N}\frac{1-\sin\big(8N\xi(\mathsf{x})\big)}{4\pi \mathsf{x}^{3/2}(1-\mathsf{x})^{1/2}} - \frac{1}{(4N)^2}\frac{(9-12\mathsf{x}+8\mathsf{x}^2)\cos\big(8N\xi(\mathsf{x})\big)}{48\pi \mathsf{x}^{2}(1-\mathsf{x})^{2}}  
        \\
        &\quad +\frac{1}{(4N)^3}
\frac{
36(3-8\mathsf x)
-
\bigl(27-72\mathsf x-288\mathsf x^2+192\mathsf x^3-64\mathsf x^4\bigr)
\sin\bigl(8N\xi(\mathsf x)\bigr)
}{
1152\pi \mathsf x^{5/2}(1-\mathsf x)^{7/2}
} + \frac{ O(N^{-4}) }{\mathsf{x}^3(1-\mathsf{x})^{5}} 
    \end{split}
    \end{equation}
    uniformly for $x\in\RN{2}$. Here, $\xi$ is given by \eqref{def xi}. 
    \smallskip 
    \item[\textup{(iii)}] \textup{\textbf{(Airy regime)}} As $N\to\infty$, we have 
    \begin{equation} \label{asymp f airy}
    \begin{aligned}
        f_N(x) &= \frac{1}{(16N)^{2/3}} \Ai(y)^2-\frac{y\,\Ai(y)}{ 10(2N)^{4/3} } \Bigl( 4\Ai(y) + y \Ai'(y) \Bigr) + \frac{1}{2800N^2}
\Big[
(362y^2+7y^5)\Ai(y)^2
\\
&\quad +
(60+146y^3)\Ai(y)\Ai'(y)
+
7y^4\Ai'(y)^2
\Big] + \textup{envAi}_{7/2}(y)^2 O(N^{-8/3})
    \end{aligned}
    \end{equation} 
    uniformly for $x\in\RN{3}$. Here, $\textup{envAi}_p$ is given by \eqref{def of envAi}. 
\smallskip 
    \item[\textup{(iv)}] \textup{\textbf{(Exponential regime)}} As $N\to\infty$, we have 
    \begin{equation} \label{asymp f exp}
        f_N(x) = \frac{1}{32\pi N}\frac{1}{\mathsf{x}^{3/2}(\mathsf{x}-1)^{1/2}} \exp\Big[-4N\Big(\sqrt{\mathsf{x}^2-\mathsf{x}}-\textup{arccosh}\sqrt{\mathsf{x}}\Big)\Big] \Big(1+O(N^{-1/4})\Big)
    \end{equation}
    uniformly for $x\in\RN{4}$.
\end{itemize}
\end{prop}

This proposition will be proved in Subsection~\ref{Subsec_proof of Prop asymp f_N}. See Figure~\ref{Fig_asymp f_N} for numerical illustrations of Proposition~\ref{prop asymp f}, which demonstrates that the asymptotic expansions remain accurate up to the displayed order. 

We also note that the asymptotic formulas in Proposition~\ref{prop asymp f} involve the quantities $\mathsf{x}$ and $\mathsf{X}$, which themselves depend on $N$. As a consequence, some terms that may initially appear to be of order one in fact contribute at lower order. For instance, the third term on the right-hand side of \eqref{asymp f bessel} is of order $O(N^{-1/2})$, due to the asymptotic behaviour $J_\nu(x)=O(x^{-1/2})$ as $x\to\infty$.

\begin{figure}[ht]
\centering

\begin{subfigure}{0.45\textwidth}
    \centering
    \includegraphics[width=\linewidth]{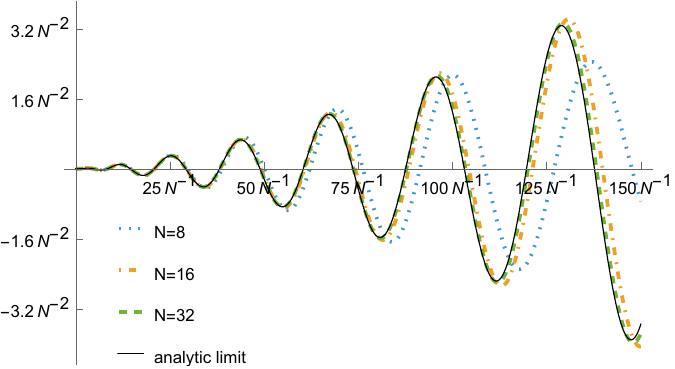}
    \caption{Regime \RN{1}}
\end{subfigure}
\hfill
\begin{subfigure}{0.45\textwidth}
    \centering
    \includegraphics[width=\linewidth]{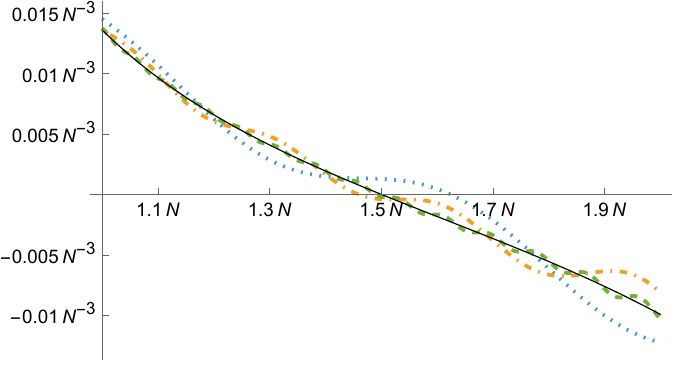}
    \caption{Regime \RN{2}}
\end{subfigure}

\vspace{0.5cm}

\begin{subfigure}{0.45\textwidth}
    \centering
    \includegraphics[width=\linewidth]{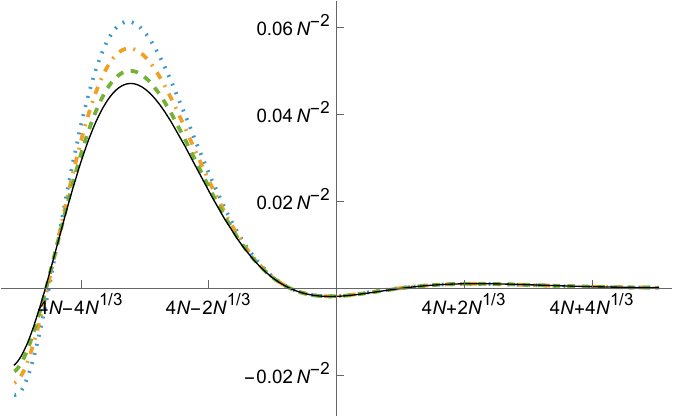}
    \caption{Regime \RN{3}}
\end{subfigure}
\hfill
\begin{subfigure}{0.45\textwidth}
    \centering
    \includegraphics[width=\linewidth]{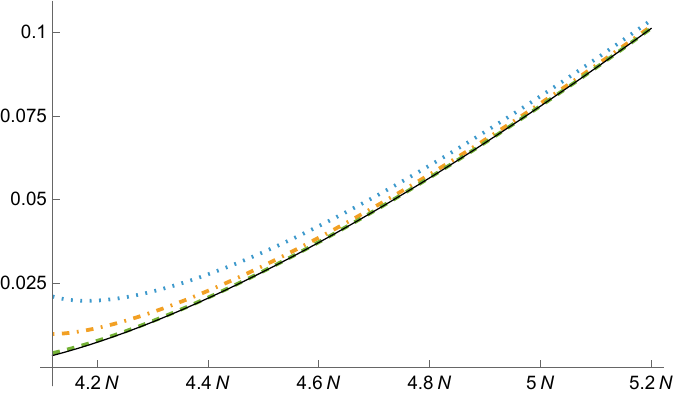}
    \caption{Regime \RN{4}}
\end{subfigure}

\caption{Each plot compares the asymptotic behaviour of $f_N$ in the corresponding regime. In the regimes $\RN{1}$, $\RN{2}$, and $\RN{3}$, the dotted, dot-dashed, and dashed curves represent the remainder obtained after subtracting the expansion up to the second term, 
for $N=8,16$, and $32$ respectively, 
while the black solid curve shows the third-order term. In the regime $\RN{2}$, the oscillatory component of the third-order term is further subtracted due to its explicit dependence on $N$. In the regime $\RN{4}$, where $f_N$ is of the form $f_N=N^{-1}g e^{-4Nh}$, the dotted, dot-dashed, and dashed curves represent $-(4N)^{-1}\log(Nf_N/g)$, whereas the black solid curve represents $h$.} \label{Fig_asymp f_N}
\end{figure}

By Proposition~\ref{Prop_BH relation}, the asymptotic behaviour of $f_N$ obtained in Proposition~\ref{prop asymp f} serves as a building block in deriving the asymptotic behaviour of $\Psi_N$.  

\begin{prop}[\textbf{Asymptotics of }$\Psi_N$] \label{prop asymp Psi}
Let $\mathsf{x}=x/4N$, $\mathsf{X}=4Nx$ and $y=(2N)^{2/3}(\mathsf{x}-1)$.
\begin{itemize}
    \item[\textup{(i)}] \textup{\textbf{(Bessel regime)}} As $N\to\infty$, we have 
    \begin{equation} \label{asymp Psi bessel}
        \Psi_N(x) = N - \frac{1}{2}\Big(\mathsf{X}J_0(\sqrt{\mathsf{X}})^2-\sqrt{\mathsf{X}}J_0(\sqrt{\mathsf{X}})J_1(\sqrt{\mathsf{X}})+\mathsf{X}J_1(\sqrt{\mathsf{X}})^2\Big) + \mathsf{X}^{3/2}O(N^{-2})
    \end{equation}
    uniformly for $x\in\RN{1}$. 
\smallskip 
    \item[\textup{(ii)}] \textup{\textbf{(Oscillatory regime)}} As $N\to\infty$, we have 
    \begin{equation} \label{asymp Psi cos}
        \Psi_N(x) = N - \frac{2N}{\pi} \Big(\sqrt{\mathsf{x}-\mathsf{x}^2}+\arcsin\sqrt{\mathsf{x}}\Big) + O(1).
    \end{equation}
    uniformly for $x\in\RN{2}$.
    \smallskip 
    \item[\textup{(iii)}] \textup{\textbf{(Airy regime)}} As $N\to\infty$, we have 
    \begin{equation} \label{asymp Psi airy}
        \Psi_N(x) = \frac{1}{3} \Big(2y^2\Ai(y)^2-\Ai(y)\Ai'(y)-2y\Ai'(y)^2\Big) + \textup{envAi}_{5/4}(y)^2 O(N^{-2/3})
    \end{equation}
    uniformly for $x\in\RN{3}$. Here, $\textup{envAi}_p$ is given by \eqref{def of envAi}. 
   \smallskip 
    \item[\textup{(iv)}] \textup{\textbf{(Exponential regime)}} As $N\to\infty$, we have 
    \begin{equation} \label{asymp Psi exp}
        \Psi_N(x) = \frac{1}{32\pi N}\frac{1}{\mathsf{x}^{1/2}(\mathsf{x}-1)^{3/2}} \exp\Big[-4N\Big(\sqrt{\mathsf{x}^2-\mathsf{x}}-\textup{arccosh}\sqrt{\mathsf{x}}\Big)\Big] \Big(1+O(N^{-1/4})\Big)
    \end{equation}
    uniformly for $x\in\RN{4}$.
\end{itemize}
\end{prop}
 
Proposition~\ref{prop asymp Psi}  will be shown in Subsection~\ref{Subsectoin_prop asymp rho Psi}.

With Propositions~\ref{prop asymp f} and~\ref{prop asymp Psi} at hand, we are now in a position to prove our main results, namely Theorems~\ref{thm strong nH}, \ref{thm mesoscopic nH}, and~\ref{thm weak nH}. 
%Although the strong, mesoscopic, and weak non-Hermiticity regimes are formulated separately, their proofs can be presented in a unified manner. For notational convenience, we regard the strong non-Hermiticity regime as corresponding to the case $\alpha=0$.

\begin{proof}[Proof of Theorems \ref{thm strong nH}, \ref{thm mesoscopic nH}, and \ref{thm weak nH}]

By Proposition \ref{prop finite N}, we have 
\begin{align}
\begin{split}
\label{eq_proof of main 1}
   \mathcal{F}_N^{\textup{(d)}} \equiv \mathcal{F}_N^{\textup{(d)}}(N^{\gamma}\mathsf{T},\theta) &= \exp\Big(-N^{2\gamma-1-\alpha}A\Big) f_N(N^{2\gamma-1}B), 
    \\
   \mathcal{F}_N^{\textup{(c)}} \equiv  \mathcal{F}_N^{\textup{(c)}}(N^{\gamma}\mathsf{T},\theta) &= N-\exp\Big(-N^{2\gamma-1-\alpha}A\Big) \Psi_N(N^{2\gamma-1}B),
\end{split}
\end{align} 
where
\begin{equation} \label{def of A B}
    A:=\begin{cases}\displaystyle
        \frac{1-\tau^2}{4}\mathsf{T}^2 &\textup{if } \alpha=0, 
        \smallskip 
        \\ \displaystyle
        \frac{\kappa}{2}\mathsf{T}^2+O(N^{-\alpha}) &\textup{if } \alpha>0,
    \end{cases}
    \qquad
    B:= \begin{cases}\displaystyle
        |\eta\mathsf{T}|^2 &\textup{if } \alpha=0, 
        \smallskip 
        \\ 
        \displaystyle \mathsf{t}^2+O(N^{-\alpha}) &\textup{if } \alpha>0.
    \end{cases}
\end{equation}
 
One observes from the exponential factor $\exp(-N^{2\gamma-1-\alpha}A)$ in \eqref{eq_proof of main 1} that the DSFF undergoes a transition to asymptotically exponential behaviour when $\gamma>\frac{1+\alpha}{2}$. Moreover, by Proposition~\ref{prop asymp f}~(iv) and Proposition~\ref{prop asymp Psi}~(iv), the quantities $f_N(N^{2\gamma-1}B)$ and $\Psi_N(N^{2\gamma-1}B)$ appearing in \eqref{eq_proof of main 1} also exhibit exponential behaviour for $\gamma>1$. Either of these two leads to the exponential decays.  
Therefore, the Heisenberg time, namely the scale at which the transition to the plateau regime occurs, is identified as $T_{\textup{H}} = O(N^{\gamma_{\textup{H}}})$, where
\begin{equation} \label{def of gammaH time}
\gamma_{\textup{H}}
:=
\begin{cases}
\displaystyle \frac{1+\alpha}{2}, & \textup{if } \alpha<1, 
\smallskip 
\\
1, & \textup{if } \alpha \ge 1.
\end{cases}
\end{equation} 
Accordingly, the two regimes $0 \le \gamma \le \gamma_{\textup{H}}$ and $\gamma>\gamma_{\textup{H}}$ are treated separately.

We first establish the asymptotic results \eqref{asymp of FN d/c snH}, \eqref{asymp of FN d/c mnH}, and \eqref{asymp of FN d/c wnH} in the dip--ramp regime corresponding to $0 \le \gamma \le \gamma_{\textup{H}}$.
Note that, by Proposition~\ref{prop asymp Psi} (i) and (ii), we have
\[ 
\Psi_N(N^{2\gamma-1}B)=N+O(N^\gamma)
\]
whenever $\delta/N \le x \ll N$ for some $\delta>0$. By combining this with \eqref{eq_proof of main 1}, we obtain that for $0\le \gamma< \gamma_\textup{H}$,
\begin{align} 
 &\mathcal{F}_N^{\textup{(d)}} = f_N(N^{2\gamma-1}B) \Big[1+O(N^{2\gamma-1-\alpha})\Big], \label{eq_proof of main 2-0} \\
    &\mathcal{F}_N^{\textup{(c)}} = N-\Psi_N(N^{2\gamma-1}B) + N^{2\gamma-\alpha}A + O(N^{\mu}),  \label{eq_proof of main 2}
\end{align} 
where
\begin{equation} \label{def of mu}
    \mu := \begin{cases} \displaystyle
        3\gamma-1-\alpha &\textup{if } \gamma<\alpha, \smallskip \\
        4\gamma-1-2\alpha &\textup{if } \gamma\ge\alpha.
    \end{cases}
\end{equation} 
Now, the case $\gamma=0$ follows directly from the asymptotic formulas \eqref{asymp f bessel} and \eqref{asymp Psi bessel}. 

On the other hand, for the case $0<\gamma<\gamma_{\textup{H}}$, we make use of the following asymptotic formulas in the regime $1/N \ll x \ll N$: 
\begin{align}
 \label{f dip ramp}
    f_N(x) &= \frac{\sqrt{N}}{\pi x^{3/2}} \bigg[\frac{1}{2}-\frac{1}{2}\sin\Big(4\sqrt{Nx}\Big)\bigg] + x^{-5/2}O(N^{-1/2}) + O(1),
    \\
   \label{Psi dip ramp}
    \Psi_N(x) &= N - \sqrt{N}\frac{2\sqrt{x}}{\pi} + \begin{cases}
        x^{-1/2}O(N^{-1/2}) & \textup{if } \frac{1}{N} \ll x \le \delta, 
        \smallskip 
        \\
        x^{3/2}O(N^{-1/2}) & \textup{if } \delta \le x \ll N,
    \end{cases} 
\end{align}
where $\delta>0$ is any positive constant. These formulas follow by combining (i) and~(ii) of Propositions~\ref{prop asymp f} and~\ref{prop asymp Psi}. 
For the regime $1/N \ll x \le \delta$, we apply the asymptotic expansion of the Bessel function \cite[Eq.~(10.17.3)]{NIST} to Propositions~\ref{prop asymp f} and~\ref{prop asymp Psi} (i). 
On the other hand, for $\delta \le x \ll N$, we use the expansion
\begin{equation*}
    \xi(x)=\sqrt{x}+O(x^{3/2}),
    \qquad x\to0,
\end{equation*}
together with Propositions~\ref{prop asymp f} and~\ref{prop asymp Psi} (ii). 

Combining \eqref{def of A B}, \eqref{eq_proof of main 2-0}, and \eqref{f dip ramp}, the desired asymptotic results \eqref{asymp of FN d/c snH}, \eqref{asymp of FN d/c mnH}, and \eqref{asymp of FN d/c wnH} for the disconnected part follow after straightforward computations. 
On the other hand, inserting \eqref{Psi dip ramp} into \eqref{eq_proof of main 2}, we obtain 
\begin{equation}
    \mathcal{F}_N^{\textup{(c)}} = N^{\gamma}\frac{2}{\pi}\sqrt{B} + N^{2\gamma-\alpha}A + O(N^{\mu}),
\end{equation}
where $\mu$ is defined as \eqref{def of mu}. 
Using again \eqref{def of A B}, the desired asymptotic results \eqref{asymp of FN d/c snH}, \eqref{asymp of FN d/c mnH}, and \eqref{asymp of FN d/c wnH} for the connected part follow by straightforward computations. 

%Then, one can see that only if $\alpha>0$, there exists the phase $\alpha>\gamma$ where the leading term is given by $N^{\gamma}\frac{2}{\pi}\sqrt{B}$. This leads to the appearance of the linear ramp under weak and mesoscopic non-Hermiticity.  

For $\gamma=\gamma_{\textup{H}}$, it follows from \eqref{def of gammaH time} that 
\begin{align}
 \mathcal{F}_N^{\textup{(d)}} 
& =\begin{cases}
  \exp(-A) f_N(N^{\alpha}B) &\textup{if }   0\le\alpha<1,
  \smallskip 
  \\
  \exp(-N^{1-\alpha}A) f_N(NB) &\textup{if } \alpha \ge 1, 
 \end{cases}
 \\
  \mathcal{F}_N^{\textup{(c)}} & =
 \begin{cases}
  N-\exp(-A) \Psi_N(N^{\alpha}B)   &\textup{if }   0\le\alpha<1,
  \smallskip 
  \\
  N-\exp(-N^{1-\alpha}A) \Psi_N(NB) &\textup{if } \alpha \ge 1. 
 \end{cases}
\end{align} 
The case $0 \le \alpha <1$ then follows from \eqref{def of A B}, \eqref{f dip ramp} and \eqref{Psi dip ramp}. For $\alpha \ge 1$, note that by the definition \eqref{def of A B}, assertions~(ii) of Propositions~\ref{prop asymp f} and~\ref{prop asymp Psi} apply in the case $\mathsf{t}<2$, yielding the desired asymptotics \eqref{asymp of FN d/c wnH}. 
On the other hand, when $\mathsf{t}>2$, assertions~(iv) of Propositions~\ref{prop asymp f} and~\ref{prop asymp Psi} apply, leading to the plateau asymptotics \eqref{asymp of FN p wnH}.

Finally, we consider the case $\gamma>\gamma_{\textup{H}}$. We establish the following estimates in the regime $x\gg1$:
\begin{equation} \label{f plateau}
    f_N(x) = \begin{cases} \displaystyle
        O(N^{1/2}) &\textup{if } 1\ll x\le \delta N,
         \smallskip 
        \\ \displaystyle
        \exp\Big(-x+O(N\log N)\Big) &\textup{if } x\gg N,
    \end{cases}
\end{equation}
and 
\begin{equation} \label{Psi plateau}
    \Psi_N(x) = \begin{cases} \displaystyle
        O(N) &\textup{if } 1\ll x\le \delta N, 
        \smallskip 
        \\ \displaystyle
        \exp\Big(-x+O(N\log N)\Big) &\textup{if } x\gg N. 
    \end{cases}
\end{equation}
for any fixed constant $\delta>0$. The uniform bound in the regime $1 \ll x \le \delta N$ appearing in \eqref{f plateau} (resp.~\eqref{Psi plateau}) follows by combining assertions~(ii), (iii), and~(iv) of Proposition~\ref{prop asymp f} (resp.,~Proposition~\ref{prop asymp Psi}), which together cover the entire regime. 
For the second regime $x \gg N$, assertion~(iv) of Proposition~\ref{prop asymp f} (resp.,~Proposition~\ref{prop asymp Psi}) yields the exponential asymptotics in \eqref{f plateau} (resp.~\eqref{Psi plateau}). 
Combining \eqref{f plateau}, \eqref{Psi plateau}, and \eqref{eq_proof of main 1}, we obtain the desired asymptotic results \eqref{asymp of FN d/c snH}, \eqref{asymp of FN d/c mnH}, and \eqref{asymp of FN d/c wnH} in the plateau regime. 
\end{proof}

\subsection{Proof of Proposition~\ref{prop asymp f}} \label{Subsec_proof of Prop asymp f_N}

In this subsection, we prove Proposition~\ref{prop asymp f}. 
By the definition \eqref{def of fN}, the key input is the asymptotic behaviour of the Laguerre polynomials. 
The basic asymptotic formulas can be found, for instance, in \cite{FW88} and \cite[Chapter 18]{NIST}. 
However, unlike the standard asymptotic expansions that are valid only in a particular region, our analysis requires uniform asymptotic expansions in each scaling regime. 
Taken together, these regimes cover the entire real axis, which is crucial for integrating the expansions and subsequently matching them in the proof of Proposition~\ref{prop asymp Psi}.  

Furthermore, we derive the expansions up to the third-order term, which is the optimal level of precision required for our method. 
The first two coefficients were obtained in \cite[Sections 4 and 5]{FW88}, whereas the third-order coefficients do not seem to have appeared in the literature. 
Nevertheless, they can be derived by following the method developed in \cite{FW88}, 
to where we refer for details. 
We also note that asymptotic expansions of this type, even at a lower level of precision order, have been used extensively in the literature, including in the recent work \cite{MS26}.

For 
%this 
these asymptotic behaviours, we define 
 \begin{equation} \label{def of zeta(x) for Airy}
        \zeta(x):=\begin{cases} \displaystyle
            -\Big(\frac{3}{4}\big(\arccos\sqrt{x}-\sqrt{x-x^2}\big)\Big)^{2/3} &\textup{if } 0\le x\le 1,
            \smallskip 
            \\ \displaystyle
            \Big(\frac{3}{4}\big(\sqrt{x^2-x}-\textup{arccosh}\sqrt{x}\big)\Big)^{2/3}  &\textup{if } x\ge1
        \end{cases}
    \end{equation}
    and
    \begin{equation} \label{def of J envelop}
        \textup{env}J_\alpha(x) = \begin{cases} \displaystyle
            J_\alpha(x) & \textup{if } 0<x<X_\alpha, 
            \smallskip 
            \\ \displaystyle
            \sqrt{J_\alpha(x)^2+Y_\alpha(x)^2} & \textup{if } x\ge X_\alpha.
        \end{cases}
    \end{equation}
Here, $X_\alpha>0$ is chosen so that $J_\alpha(x)>0$ when $0<x<X_\alpha$, and $Y_\alpha$ denotes the Bessel function of second kind or Weber's function, given by \cite[Chapter 10]{NIST}
\begin{equation}
    Y_{\alpha}(x)= \frac{J_\alpha(x)\cos(\alpha x)-J_{-\alpha}(x)}{\sin(\alpha x)},  
\end{equation}
For integer $\alpha$ it simplifies to a derivative of $J_\alpha$ with respect to its index.

\begin{lem}[\textbf{Asymptotic behaviours of Laguerre polynomials}] \label{Lem_Laguerre asymp}
Fix $0<\delta<1$ and let $\nu=4n+2\alpha+2$.
\begin{itemize}
    \item[\textup{(i)}] \textup{\textbf{(Bessel function expansion)}} 
    As $n\to\infty$,
    \begin{align}
    \begin{split} \label{asymp bessel}
      &\quad  L^{(\alpha)}_{n}(\nu x) = \frac{e^{\nu x/2} \xi(x)^{1/2}}{2^{\alpha}x^{\alpha/2+1/4}(1-x)^{1/4}} 
      \\
        & \times \bigg(J_\alpha\big(\nu\xi(x)\big) + \frac{E_1(x)}{\nu}J_{\alpha+1}\big(\nu\xi(x)\big) + \frac{E_2(x)}{\nu^2}J_{\alpha}\big(\nu\xi(x)\big) +\textup{env}J_{\alpha}\big(\nu\xi(x)\big)O(\nu^{-3}) \bigg), 
    \end{split}
    \end{align} 
    uniformly for $0\le x\le 1-\delta$.  Here, $\xi$ is given by \eqref{def xi}, $ \textup{env}J_\alpha$ is given by \eqref{def of J envelop}, and 
    \begin{equation} \label{def of E1}
        E_1(x) := \frac{4\alpha^2-1}{8\xi(x)}-\sqrt{\frac{1-x}{x}}\bigg(\frac{4\alpha^2-1}{8}+\frac{1}{4}\frac{x}{1-x}+\frac{5}{24}\Big(\frac{x}{1-x}\Big)^{2}\bigg),
    \end{equation}
    \begin{align} \label{def of E2}
    \begin{split} 
     E_2(x) &:=  -\frac{(2\alpha-1)(2\alpha+1)(2\alpha+3)(2\alpha+5)}{128\xi(x)^2} -\frac{(2\alpha-3)(2\alpha-1)(2\alpha+1)(2\alpha+3)}{128}\frac{1-x}{x} 
     \\
     & \quad + \frac{(2\alpha+1)(2\alpha+3)}{8\xi(x)}\sqrt{\frac{1-x}{x}}\bigg(\frac{4\alpha^2-1}{8}+\frac{1}{4}\frac{x}{1-x}+\frac{5}{24}\Big(\frac{x}{1-x}\Big)^{2}\bigg)
     \\
        & \quad  -\frac{(2\alpha-3)(2\alpha-1)(8\alpha+7)}{96}+\Big(\frac{7\alpha^2}{48}-\frac{121}{192}\Big)\frac{x}{1-x}  -\frac{77}{96}\Big(\frac{x}{1-x}\Big)^2-\frac{385}{1152}\Big(\frac{x}{1-x}\Big)^3. 
    \end{split}
    \end{align} 

    \item[\textup{(ii)}] \textup{\textbf{(Airy function expansion)}} 
    As $n\to\infty$,
    \begin{align}
    \begin{split}
     \label{asymp airy}
     &\quad   L^{(\alpha)}_{n}(\nu x)   = \frac{(-1)^n}{\nu^{1/3}} \frac{e^{\nu x/2}}{2^{\alpha-1/2}x^{\alpha/2+1/4}} \Big(\frac{\zeta(x)}{x-1}\Big)^{1/4}
        \\
        & \times \bigg( \Ai\big(\nu^{2/3}\zeta(x)\big) + \frac{F_1(x)}{\nu^{4/3}}\Ai'\big(\nu^{2/3}\zeta(x)\big)  
      + \frac{F_2(x)}{\nu^{2}}\Ai\big(\nu^{2/3}\zeta(x)\big) + \textup{envAi}_{1/4}\big(\nu^{2/3}\zeta(x)\big) \,O(\nu^{-10/3}) \bigg),  
    \end{split}
    \end{align} 
    uniformly for $\delta\le x<\infty$. Here, $\zeta$ is given by \eqref{def of zeta(x) for Airy}, $\textup{envAi}_p$ is given by \eqref{def of envAi}, and 
    \begin{equation}
        F_1(x) := -\frac{5}{48} \zeta(x)^{-2} +\sqrt{\frac{x-1}{x\zeta(x)}} \Big[\frac{4\alpha^2-1}{8} - \frac{1}{4}\Big(\frac{x}{x-1}\Big)+\frac{5}{24}\Big(\frac{x}{x-1}\Big)^{2}\Big],
    \end{equation}
    \begin{align}
    \begin{split}
        F_2(x) & :=  -\frac{455}{4608\zeta(x)^3} + \frac{7}{48\zeta(x)}\sqrt{\frac{x-1}{x\zeta(x)}}\bigg(\frac{4\alpha^2-1}{8}-\frac{1}{4}\frac{x}{x-1}+\frac{5}{24}\Big(\frac{x}{x-1}\Big)^{2}\bigg) 
        \\
        & \quad +\frac{(2\alpha-3)(2\alpha-1)(2\alpha+1)(2\alpha+3)}{128}\frac{x-1}{x} -\frac{(2\alpha-3)(2\alpha-1)(8\alpha+7)}{96}
        \\
        & \quad +\Big(\frac{121}{192}-\frac{7\alpha^2}{48}\Big)\frac{x}{x-1} -\frac{77}{96}\Big(\frac{x}{x-1}\Big)^2+\frac{385}{1152}\Big(\frac{x}{x-1}\Big)^3. 
    \end{split}
    \end{align} 
\end{itemize}
\end{lem}

As previously mentioned, Lemma~\ref{Lem_Laguerre asymp} follows from~\cite{FW88}, where the explicit coefficients up to \(E_1\) and \(F_1\) are presented. The method developed in~\cite{FW88} also applies to obtain more refined asymptotic expansions. For the reader's convenience, we outline the corresponding computations in Appendix~\ref{Appendix_Laguerre coeff}. 

\begin{rem}
In the Airy function expansion~\eqref{asymp airy}, the error term is expressed in terms of the envelope Airy function $\textup{envAi}_p$ defined in~\eqref{def of envAi}. This definition differs from the standard one appearing in the literature; see e.g. \cite[Eqs.~(2.8.20)--(2.8.21)]{NIST} and \cite[Eq.~(5.17)]{FW88}. Nevertheless, it follows from the asymptotic behaviour of the Airy function~\cite[Eq.~(9.7.5)]{NIST} that the two definitions are equivalent. The formulation~\eqref{def of envAi} is, however, more convenient for our analysis. 
\end{rem}

By using Lemma~\ref{Lem_Laguerre asymp}, we now prove Proposition~\ref{prop asymp f}.  

\begin{proof}[Proof of Proposition~\ref{prop asymp f}] Recall that $f_N$ is defined in \eqref{def of fN}. 
We also recall that the regions $\RN{1}$, $\RN{2}$, $\RN{3}$, and $\RN{4}$ are introduced in \eqref{def of regions I II III IV}. 

We first show \eqref{asymp f bessel}. For $x\in\RN{1}$, we apply the asymptotic expansion \eqref{asymp bessel}. Note that as $\mathsf{x}=x/4N \to 0$, it follows from \eqref{def xi}, \eqref{def of E1} and \eqref{def of E2} that
$$\xi(\mathsf{x})=\sqrt{\mathsf{x}}-\frac{1}{6}x^{\mathsf{x}/2}-\frac{1}{40}x^{\mathsf{x}/2}+O(\mathsf{x}^{7/2}), \qquad E_1(\mathsf{x})=-\frac{4}{15}\mathsf{x}^{3/2}+O(\mathsf{x}^{5/2}), \qquad E_2(\mathsf{x})=O(\mathsf{x}^2).$$ Then, for $x\in\RN{1}$ and $\mathsf{X}=4Nx$, by \eqref{asymp bessel}, we obtain 
\begin{align*}
e^{-x/2} L_{N-1}^{(1)}(x)  
    & = 2N \frac{J_1(\sqrt{\mathsf{X}})}{\sqrt{\mathsf{X}}} + \frac{1}{48N} \mathsf{X} J_2(\sqrt{\mathsf{X}}) + \frac{\mathsf{X}}{46080N^3}\Big[\sqrt{\mathsf{X}}(24-5\mathsf{X})J_1(\sqrt{\mathsf{X}})-(96-24\mathsf{X})J_2(\sqrt{\mathsf{X}})\Big] \\
    &\quad+ \mathsf{X}^{11/4}O(N^{-5}) + O(N^{-2})
\end{align*} 
as $N \to \infty.$ 
Here, we have also used the recurrence relation (see e.g. \cite[Eqs. (10.6.1)--(10.6.2)]{NIST})
\[
J_1(x)-xJ_1'(x)=xJ_2(x),
\]
and the global bound
\begin{equation*}
    x^pJ_{\alpha}(x)=O(x^{p-1/2})
\end{equation*}
for any $p\in\mathbb{R}$ and $\alpha \ge 0$, which follows from the standard asymptotic behaviours of the Bessel function \cite[Eqs. (10.7.3), (10.17.3)]{NIST}. Then, by straightforward computation, we obtain \eqref{asymp f bessel}.

Next, we establish~\eqref{asymp f cos}. To treat the regime~$\RN{2}$, we employ both the Bessel expansion~\eqref{asymp bessel} and the Airy expansion~\eqref{asymp airy}. For $x\in \RN{2}\cap[0,4N(1-\delta)]$, the argument of the Bessel function in~\eqref{asymp bessel} satisfies
\[
4N\xi(\mathsf{x}) \ge C\sqrt{N},
\]
for some constant $C>0$. Hence the Bessel function admits an asymptotic expansion in terms of trigonometric functions; see~\cite[Eq.~(10.17.3)]{NIST}. Then it follows that 
\begin{align*}
  e^{-x/2} L_{N-1}^{(1)}(x) &= \frac{ \mathsf{x}^{-3/4}(1-\mathsf{x})^{-1/4} }{\sqrt{8\pi N}} \bigg[ \cos\Big(4N\xi(\mathsf{x})-\frac{3}{4}\pi\Big) -\frac{1}{4N}\Big(\frac{3}{8\xi(\mathsf{x})}-E_1(\mathsf{x})\Big)\sin\Big(4N\xi(\mathsf{x})-\frac{3}{4}\pi\Big) \\
    &\quad +\frac{1}{(4N)^2}\Big(\frac{15}{128\xi(\mathsf{x})^2}+\frac{15E_1(\mathsf{x})}{8\xi(\mathsf{x})}+E_2(\mathsf{x})\Big)\cos\Big(4N\xi(\mathsf{x})-\frac{3}{4}\pi\Big) +\xi(\mathsf{x})^{-3}O(N^{-3}) \bigg].
\end{align*} 
Therefore, we obtain~\eqref{asymp f cos} for $x\in \RN{2}\cap[0,4N(1-\delta)]$, with error term of order $\mathsf{x}^{-3}O(N^{-4})$. 

It remains to verify that~\eqref{asymp f cos} also holds for $x\in \RN{2}\cap[4N(1-\delta),\infty)$, with error term $(1-\mathsf{x})^{-5}O(N^{-4})$. To this end, we show that the Airy function expansion~\eqref{asymp airy} yields the same asymptotic behaviour for $x\in \RN{2}\cap[4N\delta,\infty)$ together with the desired error estimate. Observe that for $x\in \RN{2}\cap[4N\delta,\infty)$, one has
\[
(4N)^{2/3}\zeta(x)\to -\infty,
\]
so that the trigonometric expansion of the Airy function given in~\cite[Eq.~(9.7.9)]{NIST} is applicable. Using the identity 
\begin{equation*}
    -\frac{2}{3}\big(-\zeta(x)\big)^{3/2} = \xi(x)-\frac{\pi}{4}, \qquad x\in[0,1],
\end{equation*}
relating \eqref{def xi} and \eqref{def of zeta(x) for Airy},  
we have
\begin{align*}
e^{-x/2} L_{N-1}^{(1)}(x) 
    & = \frac{ \mathsf{x}^{-3/4}(1-\mathsf{x})^{-1/4} }{\sqrt{8\pi N}}  \bigg[ \cos\Big(4N\xi(\mathsf{x})-\frac{3}{4}\pi\Big) -\frac{\big(-\zeta(\mathsf{x})\big)^{1/2}}{4N}\Big(\frac{5}{48\zeta(\mathsf{x})^{2}}+F_1(\mathsf{x})\Big)\sin\Big(4N\xi(\mathsf{x})-\frac{3}{4}\pi\Big) \\
    &\quad + \frac{1}{(4N)^2}\Big(\frac{385}{4608\zeta(\mathsf{x})^3}-\frac{7F_1(\mathsf{x})}{48\zeta(\mathsf{x})}+F_2(\mathsf{x})\Big)\cos\Big(4N\xi(\mathsf{x})-\frac{3}{4}\pi\Big) +(-\zeta(\mathsf{x}))^{-9/2}O(N^{-3}) \bigg].
\end{align*} 
Since $\zeta(x)=O(|x-1|)$ for $x\in[0,1]$, this yields~\eqref{asymp f cos} for $x\in \RN{2}\cap[4N\delta,\infty)$, with the required error bound $(1-\mathsf{x})^{-5}O(N^{-4})$. This completes the proof of~\eqref{asymp f cos}.

For $x\in \RN{3}$, the argument of the Airy function in~\eqref{asymp airy} admits the expansion
\begin{equation*}
    (4N)^{2/3}\zeta(\mathsf{x})= y-\frac{y^2}{5(2N)^{2/3}}+\frac{17y^3}{175(2N)^{2/3}}+y^4O(N^{-2}).
\end{equation*}
Substituting this into~\eqref{asymp airy}, we obtain 
\begin{align*}
 e^{-x/2} L_{N-1}^{(1)}(x)  & = \frac{(-1)^{N-1}}{(16N)^{1/3}} \bigg[ \Ai(y) - (2N)^{-2/3}\Big(\frac{4}{5}y\Ai(y)+\frac{1}{5}y^2\Ai'(y)\Big) 
 \\
    &\quad  + (2N)^{-4/3}\bigg(\Big(\frac{5}{7}y^2+\frac{1}{50}y^5\Big)\Ai(y)  +\Big(\frac{9}{35}y^3+\frac{6}{35}\Big)\Ai'(y)\bigg) + \text{envAi}_{15/2}(y) O(N^{-8/3}) \bigg],
\end{align*}
which leads to \eqref{asymp f airy}.

Finally, for $x\in \RN{4}$, we apply the exponential asymptotics of the Airy function; see~\cite[Eqs.~(9.7.5)--(9.7.6)]{NIST}. Together with the expansion
\begin{equation*}
    \zeta(x)=\frac{1}{2^{2/3}}(x-1)+O(|x-1|^2),
\end{equation*}
we obtain 
\begin{equation*}
    e^{-x/2} L_{N-1}^{(1)}(x) = \frac{(-1)^{N-1}}{(32\pi N)^{1/2}} \frac{1}{\mathsf{x}^{3/4}(\mathsf{x}-1)^{1/4}} \exp\Big[-2N\Big(\sqrt{\mathsf{x}^2-\mathsf{x}}-\textup{arccosh}\sqrt{\mathsf{x}}\Big)\Big] \Big(1+O(N^{-\frac{1}{4}})\Big)
\end{equation*}
which leads to \eqref{asymp f exp}. This completes the proof. 
\end{proof}

\subsection{Proof of Proposition \ref{prop asymp Psi}} \label{Subsectoin_prop asymp rho Psi}

In this subsection, we prove Proposition~\ref{prop asymp Psi}. Recall that $\Psi_N$ is given by \eqref{def of PsiN}. 
By Proposition~\ref{Prop_BH relation}, the asymptotic behaviour of $\Psi_N$ is essentially obtained by integrating $f_N$ twice. As an intermediate step, we first derive the asymptotic behaviour of $\rho_N$ in \eqref{def of rho}, which arises from a single integration of $f_N$. Since $\rho_N$ can be identified with the density of the %Laguerre unitary ensemble (LUE), 
LUE, 
these asymptotics are also of independent interest.

\begin{prop}[\textbf{Asymptotics of }$\rho_N$] \label{prop asymp rho}
Let $\mathsf{x}=x/4N$, $\mathsf{X}=4Nx$ and $y=(2N)^{2/3}(\mathsf{x}-1)$.
\begin{itemize}
    \item[\textup{(i)}] \textup{\textbf{(Bessel regime)}} As $N\to\infty$, we have
    \begin{equation} \label{asymp rho bessel}
    \begin{aligned}
        \rho_N(x) &=  N \Big( J_0(\sqrt{\mathsf{X}})^2+J_1(\sqrt{\mathsf{X}})^2 \Big) -\frac{1}{48N}
\Bigl( 2\mathsf XJ_0(\sqrt{\mathsf X})^2 -4\sqrt{\mathsf X}J_0(\sqrt{\mathsf X})J_1(\sqrt{\mathsf X}) +\mathsf XJ_1(\sqrt{\mathsf X})^2
\Bigr) + O(N^{-1})
% \\
% & \quad +\frac{\sqrt{\mathsf X}}{23040\,N^3}
% \Big[
% \sqrt{\mathsf X}(96-26\mathsf X)J_0(\sqrt{\mathsf X})^2+
% \sqrt{\mathsf X}(6\mathsf X-56)J_1(\sqrt{\mathsf X})^2
% \\
% &\quad - (192-104\mathsf X+5\mathsf X^2)
% J_0(\sqrt{\mathsf X})J_1(\sqrt{\mathsf X})
% \Big]+ O(N^{-1})
    \end{aligned}
    \end{equation}
    uniformly for $x\in\RN{1}$. 
\smallskip 
    \item[\textup{(ii)}] \textup{\textbf{(Oscillatory regime)}} As $N\to\infty$, we have 
    \begin{equation} \label{asymp rho cos}
    \begin{aligned}
        \rho_N(x)  =  \frac{1}{2\pi} \sqrt{\frac{1-\mathsf{x}}{\mathsf{x}}} -\frac{1}{16\pi N}\frac{\cos\big(8N\xi(\mathsf{x})\big)}{\mathsf{x}(1-\mathsf{x})} + O(N^{-1})
    \end{aligned}
    \end{equation}
    uniformly for $x\in\RN{2}$. Here, $\xi$ is given by \eqref{def xi}. 
\smallskip 
    \item[\textup{(iii)}] \textup{\textbf{(Airy regime)}} As $N\to\infty$, we have 
    \begin{equation} \label{asymp rho airy}
        \rho_N(x) = \frac{1}{(16N)^{1/3}} \Big(\Ai'(y)^2-y\Ai(y)^2\Big)
         + \textup{envAi}_{3/4}(y)^2 O(N^{-1})
    \end{equation}
    uniformly for $x\in\RN{3}$. Here, $\textup{envAi}_p$ is given by \eqref{def of envAi}. 
\smallskip 
    \item[\textup{(iv)}] \textup{\textbf{(Exponential regime)}} As $N\to\infty$, we have 
    \begin{equation} \label{asymp rho exp}
        \rho_N(x) = \frac{1}{32\pi N}\frac{1}{\mathsf{x}(\mathsf{x}-1)} \exp\Big[-4N\Big(\sqrt{\mathsf{x}^2-\mathsf{x}}-\textup{arccosh}\sqrt{\mathsf{x}}\Big)\Big] \Big(1+O(N^{-1/4})\Big)
    \end{equation}
    uniformly for $x\in\RN{4}$.
\end{itemize}
\end{prop}
 
\begin{rem}[LUE density] \label{Rem_LUE density}
As previously mentioned, the function $\rho_N$ can be identified with the one-point function of the LUE. Thus Proposition~\ref{prop asymp rho} describes the large-$N$ asymptotic behaviour of the LUE average density across all scaling regimes. In the literature, the regimes $\RN{1}$, $\RN{2}$, $\RN{3}$, and $\RN{4}$ are also commonly referred to as the hard edge, bulk, soft edge, and exterior regimes, respectively, and the corresponding asymptotics have been extensively studied for general rectangular parameter. In particular, the leading term in \eqref{asymp rho cos} is given by the Marchenko--Pastur law, while the correction term coincides with \cite[Eq.~(52)]{GFF05}. At the soft and hard edges, it was previously shown in \cite[Eq.~(73)]{GFF05} and \cite[Eq.~(4.27)]{FT19}, respectively, that the asymptotics \eqref{asymp rho airy} and \eqref{asymp rho bessel} hold up to the second-order term. 

However, these previously known results are not sufficient for our purposes, since they are formulated only for fixed scaling limits, and simply combining them does not fully cover the domains required in the present analysis. For instance, \cite[Eq.~(52)]{GFF05} states that \eqref{asymp rho cos} holds for $x=4N\mathsf{x}$ with fixed $\mathsf{x}\in(0,1)$, whereas we prove that it holds uniformly for $x\in \RN{2}=[c,4N-d\sqrt{N}]$ with fixed $c,d>0$. 
\end{rem}

To prove Proposition~\ref{prop asymp rho}, we distinguish between the different regimes. For $x\in \RN{1}\cup\RN{2}$, we use the representation
\begin{equation}
    \rho_N(x)
    =
    N-\int_0^x f_N(u)\,d u,
\end{equation}
whereas for $x\in \RN{3}\cup\RN{4}$, we employ
\begin{equation}
    \rho_N(x)  =  \int_x^\infty f_N(u)\,d u.
\end{equation}
%For each of the four regimes, the integral involves only elementary calculus and differential relations of the Bessel and Airy functions. 
We first show the following elementary lemma.

\begin{lem} \label{integral lemma}
Suppose that $\alpha,\beta,c,d>0$ are positive constants, and let $p_1(x)$ and $q_1(x)$ be polynomials. Define
\begin{equation} \label{def of varphi}
    \varphi(x) := \begin{cases} \displaystyle
        \sqrt{x-x^2}+\textup{arcsin}\sqrt{x} & \textup{ if } 0<x<1, 
        \smallskip 
        \\ \displaystyle
        \sqrt{x^2-x}-\textup{arccosh}\sqrt{x} & \textup{ if } x\ge1.
    \end{cases}
\end{equation}
Then we have the following. 
\begin{itemize}
    \item[(i)] For sufficiently large $\lambda>0$, suppose that $c/\lambda\le a\le b\le 1-d\lambda^{-1/2}$, and let $M$ be an arbitrary positive integer. Then we have 
    \begin{equation}
        \int^b_a \frac{p_1(x)}{x^{\alpha}(1-x)^{\beta}} e^{i\lambda\varphi(x)} \,dx = F(b)-F(a),
    \end{equation}
    where 
    \begin{equation*}
        F(x) = -\sum_{j=1}^{M} \Big(\frac{2i}{\lambda}\Big)^{j} \frac{p_j(x)}{x^{\alpha+(j-2)/2}(1-x)^{\beta+(3j-2)/2}} e^{i\lambda\varphi(x)} + O\Big(\frac{1}{\lambda^{M+1}}\Big) \frac{1}{x^{\alpha+(M-1)/2}(1-x)^{\beta+(3M+1)/2}}.
    \end{equation*}
   Here, the sequence of polynomials $(p_j)_{j\ge1}$ is defined recursively by 
    \begin{equation} \label{def of pj}
        p_{j+1}(x) := x^{\alpha+j/2}(1-x)^{\beta+3j/2} \frac{d}{dx}\Big(\frac{p_{j}(x)}{x^{\alpha+(j-2)/2}(1-x)^{\beta+(3j-2)/2}}\Big), \qquad j\ge1.
    \end{equation}

    \item[(ii)] For sufficiently large $\lambda>0$, suppose that $1+d\lambda^{-1/2}\le a\le b$, and let $M$ be an arbitrary positive integer. Then we have
    \begin{equation}
        \int^b_a \frac{q_1(x)}{x^{\alpha}(x-1)^{\beta}} e^{i\lambda\varphi(x)} \,dx = G(b)-G(a),
    \end{equation}
    where 
    \begin{equation*}
        G(x) = \sum_{j=1}^{M} \frac{(-1)^{j-1}}{\lambda^j} \frac{q_j(x)}{x^{\alpha+(j-2)/2}(x-1)^{\beta+(3j-2)/2}} e^{\lambda\varphi(x)} + O\Big(\frac{1}{\lambda^{M+1}}\Big) \frac{1}{x^{\alpha+(M-1)/2}(x-1)^{\beta+(3M+1)/2}}.
    \end{equation*}
    Here, the sequence of polynomials $(q_j)_{j\ge1}$ is defined recursively by 
    \begin{equation} \label{def of qj}
        q_{j+1}(x) := x^{\alpha+j/2}(x-1)^{\beta+3j/2} \frac{d}{dx}\Big(\frac{q_{j}(x)}{x^{\alpha+(j-2)/2}(x-1)^{\beta+(3j-2)/2}}\Big), \qquad j\ge1.
    \end{equation}
\end{itemize}
\end{lem}

\begin{proof}
We prove only~(i), since the proof of~(ii) follows by the same argument. Note that by \eqref{def of varphi}, we have
\begin{equation*}
    \frac{d}{dx}\varphi(x) = \begin{cases} \displaystyle
        x^{-1/2}(1-x)^{1/2} & \textup{ if } 0<x<1, 
        \smallskip 
        \\ 
        \displaystyle
        x^{-1/2}(x-1)^{1/2} & \textup{ if } x\ge1.
    \end{cases}
\end{equation*}
Then, repeated integration by parts yields 
\begin{equation*}
\begin{split}
    \int^b_a \frac{p_1(x)}{x^{\alpha}(1-x)^{\beta}} e^{i\lambda\varphi(x)} \,dx
    &=  -\sum_{j=1}^{M} \Big(\frac{i}{\lambda}\Big)^{j} \frac{p_j(x)}{x^{\alpha+(j-2)/2}(1-x)^{\beta+(3j-2)/2}} e^{i\lambda\varphi(x)}\bigg\vert^b_a 
    \\
    &\quad + \Big(\frac{i}{\lambda}\Big)^{M} \int_a^b \frac{p_{M+1}(x)}{x^{\alpha+M/2}(1-x)^{\beta+3M/2}} e^{i\lambda\varphi(x)} \,dx, 
\end{split}
\end{equation*}
where $(p_j)_{j\ge1}$ is given by \eqref{def of pj}. Thus, it suffices to show the error bound
\begin{equation*}
    E(x) := \Big(\frac{i}{\lambda}\Big)^{M} \int_a^b \frac{p_{M+1}(x)}{x^{\alpha+M/2}(1-x)^{\beta+3M/2}} e^{i\lambda\varphi(x)} \,dx = \frac{1}{x^{\alpha+(M-1)/2}(1-x)^{\beta+(3M+1)/2}}O(\lambda^{-M-1}).
\end{equation*}
Applying integration by parts to $E(x)$ once more, for a positive integer $m$, we obtain 
\begin{align*}
    E(x) = &-\sum_{j=M+1}^{M+m} \Big(\frac{i}{\lambda}\Big)^{j} \frac{p_j(x)}{x^{\alpha+(j-2)/2}(1-x)^{\beta+(3j-2)/2}} e^{i\lambda\varphi(x)}\bigg\vert^b_a \\
    &+ \Big(\frac{i}{\lambda}\Big)^{M+m} \int_a^b \frac{p_{M+m}(x)}{x^{\alpha+(M+m)/2}(1-x)^{\beta+3(M+m)/2}} e^{i\lambda\varphi(x)} \,dx. 
\end{align*}
It is immediate that the first term on the right-hand side satisfies the desired bound. For the second term, since $x\in [a,b]\subset [c/\lambda,1-d\lambda^{-1/2}]$, we obtain 
\begin{equation*}
\begin{split}
    &\quad \bigg\vert \Big(\frac{2i}{\lambda}\Big)^{M+m} \int_a^b \frac{p_{M+m}(x)}{x^{\alpha+(M+m)/2}(1-x)^{\beta+3(M+m)/2}} e^{i\lambda\varphi(x)} \,dx \bigg\vert \\
    & = O(\lambda^{-M-m})\sup_{x\in[a,b]} \frac{1}{x^{\alpha+(M+m)/2}(1-x)^{\beta+3(M+m)/2}} = O(\lambda^{\max\{\alpha-(M+m)/2,\beta/2-(M+m)/4\}}).
\end{split}
\end{equation*}
Choosing $m$ sufficiently large, the above term is also of order $O(\lambda^{-M-1})$. This completes the proof. 
\end{proof}

\begin{proof}[Proof of Propositions~\ref{prop asymp Psi} and ~\ref{prop asymp rho}] 
Note that by Proposition~\ref{Prop_BH relation}, Propositions~\ref{prop asymp rho} and~\ref{prop asymp Psi} are obtained by integrating the asymptotic expansion in Proposition~\ref{prop asymp f} once and twice, respectively. The main issue is to control the error terms appropriately, and the necessary estimates are already established in Proposition~\ref{prop asymp f}. Since the proof of Proposition~\ref{prop asymp Psi} follows along the same lines, differing only in the explicit form of the functions involved, we present the details only for Proposition~\ref{prop asymp rho}. 

Throughout the proof, we write $f_N^{\RN{1}}$, $f_N^{\RN{2}}$, $f_N^{\RN{3}}$, and $f_N^{\RN{4}}$ for the limiting profile functions appearing on the right-hand sides of  \eqref{asymp f bessel}--\eqref{asymp f exp}. For example, for $x\in\RN{1}$, we have 
\begin{equation*}
    f_N(x) = f_N^{\RN{1}}(\mathsf{X})+O(N^{-1}), 
\end{equation*}
where
\begin{align*}
 f_N^{\RN{1}}(u) & := 4N^2\frac{J_1(\sqrt{u})^2}{u} + \frac{1}{12}\sqrt{u}J_1(\sqrt{u})J_2(\sqrt{u})
    \\
    &\quad  + \frac{1}{11520N^2}\Big[u(24-5u)J_1(\sqrt{u})^2 -24\sqrt{u}(4-u)J_1(\sqrt{u})J_2(\sqrt{u})+5u^2J_2(\sqrt{u})^2\bigg].
\end{align*} 

We first establish \eqref{asymp rho bessel}. For $x\in\RN{1}$, direct integration of \eqref{asymp f bessel}, together with the differential identities for Bessel functions (see e.g. \cite[Eqs.~(10.6.1)--(10.6.2)]{NIST}), yields 
\begin{align} \label{eq rho 1}
    \rho_N(x) = N - \frac{1}{4N}\int_0^{\mathsf{X}} f_N^{\RN{1}}(u) \,du + O(N^{-1}) = \rho_N^{\RN{1}}(\mathsf{X}) + O(N^{-1}),
\end{align}
where 
\begin{align*}
    \rho_N^{\RN{1}}(u) &:= N \Big( J_0(\sqrt{u})^2+J_1(\sqrt{u})^2 \Big) - \frac{1}{48N}\Big(2uJ_0(\sqrt{u})^2-4\sqrt{u}J_0(\sqrt{u})J_1(\sqrt{u})+uJ_1(\sqrt{u})^2\Big) 
    \\
    &\quad + \frac{1}{23040N^3}\Big(2u(48-13u)J_0(\sqrt{u})^2 -\sqrt{u}(192-104u+5u^2)J_0(\sqrt{u})J_1(\sqrt{u}) -2u(28-3u)J_1(\sqrt{u})^2\Big). 
\end{align*}
This gives rise to \eqref{asymp rho bessel}.

Next, we establish \eqref{asymp rho cos}. For $x\in\RN{2}$, integrating \eqref{asymp f cos} and applying Lemma~\ref{integral lemma}~(i), we obtain 
\begin{align*}
  \rho_N(x) & = \rho_N(c) - 4N \int_{c/4N}^{\mathsf{x}} f_N^{\RN{2}}(u) \,du + O(N^{-3}) \int_{c/4N}^{\mathsf{x}} \frac{1}{u^{3}(1-u)^{5}} \,du \\
    & = \rho_N^{\RN{2}}(\mathsf{x}) -\rho_N^{\RN{2}}(c/4N) + \rho_N(c) + O(N^{-1}),
\end{align*} 
where 
\begin{align*}
 \rho_N^{\RN{2}}(u) &:= \frac{1}{2\pi} \sqrt{\frac{1-u}{u}} -\frac{1}{16\pi N}\frac{\cos\big(8N\xi(u)\big)}{u(1-u)} + \frac{1}{256\pi N^2} \bigg(\frac{1}{u^{3/2}(1-u)^{5/2}} -\frac{(3-12u-8u^2)\sin\big(8N\xi(u)\big)}{3u^{3/2}(1-u)^{5/2}}\bigg).
\end{align*}
Using \eqref{eq rho 1}, we observe that the boundary terms cancel: 
\begin{equation*}
    -\rho_N^{\RN{2}}(c/4N) + \rho_N(c) = -\rho_N^{\RN{2}}(c/4N) + \rho_N^{\RN{1}}(4Nc) + O(N^{-1}) = O(N^{-1}).
\end{equation*}
Moreover, the $O(N^{-1})$ error dominates the third expansion term of $\rho_N^{\RN{2}}$, which yields \eqref{asymp rho cos}. 

We now turn to the proof of \eqref{asymp rho exp}. For $x\in\RN{4}$, integrating \eqref{asymp f exp} and applying Lemma~\ref{integral lemma}~(ii), we obtain 
\begin{equation} \label{eq rho 4}
    \rho_N(x) = \frac{1}{4N} \int_{\mathsf{x}}^{\infty} f_N^{\RN{4}}(u) \,du \Big(1+O(N^{-1/4})\Big) = \rho_N^{\RN{4}}(\mathsf{x}) \Big(1+O(N^{-1/4})\Big),
\end{equation}
where 
\begin{align*}
    \rho_N^{\RN{4}}(u) &:= \frac{1}{32\pi N}\frac{1}{u(u-1)} \exp\Big[-4N\Big(\sqrt{u^2-u}-\textup{arccosh}\sqrt{u}\Big)\Big].
\end{align*}
This gives rise to the desired asymptotic formula \eqref{asymp rho exp}. 

Finally, we establish \eqref{asymp rho airy}. For $x\in\RN{3}$, using \eqref{asymp f airy}, we obtain
\begin{align}\label{eq rho 3}
\begin{split}
    \rho_N(x) &= \rho_N\big(4N+d\sqrt{N}\big) + (16N)^{1/3} \int_{y}^{dN^{1/6}/2^{4/3}} f_N^{\RN{3}}(u) \,du + O(N^{-7/3}) \int_{y}^{dN^{1/6}/2^{4/3}} \textup{envAi}_{7/2}(u)^2 \,du \\
    &= \rho_N^{\RN{3}}(y)-\rho_N^{\RN{3}}\Big(\frac{dN^{1/6}}{2^{4/3}}\Big) + \rho_N\big(4N+d\sqrt{N}\big) + \textup{envAi}_{4}(y)^2 O(N^{-7/3}),
\end{split}
\end{align} 
where
\begin{equation*}
\begin{split}
    \rho_N^{\RN{3}}(u) &:= \frac{1}{(16N)^{1/3}} \Big(\Ai'(u)^2-u\Ai(u)^2\Big) +\frac{1}{20N} \Big(3u^2\Ai(u)^2+2\Ai(u)\Ai'(u)-2u\Ai'(u)^2\Big) \\
    &\qquad +\frac{1}{350(2N)^{5/3}} \Big((7-96u^3)\Ai(u)^2-u(74+7u^3)\Ai(u)\Ai'(u)+37u^2\Ai'(u)^2\Big).
\end{split}
\end{equation*}
Here, the integral of $f_N^{\RN{3}}$ follows from the differential equation $\Ai''(u)=u\Ai(u)$, while the error term is estimated using 
\begin{equation*}
    \int\textup{envAi}_p(x)^2\,dx=\textup{envAi}_{p+1/2}(x)^2,
\end{equation*}
which follows directly from the definition \eqref{def of envAi}. Then, one sees from the asymptotics of the Airy function \cite[Eqs. (9.7.5)--(9.7.6), (9.7.9)--(9.7.10)]{NIST} that
\begin{equation*}
    \rho_N^{\RN{3}}(y) + \textup{envAi}_{4}(y)^2 O(N^{-7/3}) = \frac{1}{(16N)^{1/3}} \Big(\Ai'(y)^2-u\Ai(y)^2\Big) + \textup{envAi}_{3/4}(y)^2O(N^{-1}).
\end{equation*}
Then, using \eqref{eq rho 4}, we obtain the cancellation of the boundary terms:
\begin{align*}
    \rho_N(4N+d\sqrt{N})-\rho_N^{\RN{3}}\Big(\frac{dN^{1/6}}{2^{4/3}}\Big)
    = O(N^{-3/4})\exp\Big[-\frac{1}{3}d^{3/2}N^{1/4}\Big].
\end{align*}
Combining this with \eqref{eq rho 3}, we obtain \eqref{asymp rho airy}, which completes the proof.  
\end{proof}

\appendix

\section{Alternative representations of 
$\Psi_N$}

In this appendix, we present alternative representations of the function $\Psi_N$ in \eqref{def of PsiN}. For this, we begin with the following elementary lemma. 

\begin{lem} \label{Lem_Laguerre sum diff identity}
We have 
\begin{equation}
\frac{d}{dx} \bigg[ e^{-x} \sum_{k=0}^{N-1} L_{k}^{(\nu)}(x)^2  \bigg] = -e^{-x} L_{N-1}^{(\nu+1)}(x)^2. 
\end{equation}
\end{lem}
\begin{proof} 

By using \eqref{derivative of Laguerre} and  the elementary parameter-shift identity
\begin{equation}\label{shift of Laguerre}
L_k^{(\nu)}(x)=L_k^{(\nu+1)}(x)-L_{k-1}^{(\nu+1)}(x)\qquad (k\ge 0), 
\end{equation} 
with the convention $L_{-1}^{(\nu+1)}\equiv 0$, we have  
\begin{align*}
\frac{d}{dx} \bigg[ e^{-x} \sum_{k=0}^{N-1} L_{k}^{(\nu)}(x)^2  \bigg] & = -e^{-x}\sum_{k=0}^{N-1}L_k^{(\nu)}(x)\Bigl[ L_k^{(\nu)}(x) +2\,L_{k-1}^{(\nu+1)}(x) \Bigr]
\\
&= -e^{-x}\sum_{k=0}^{N-1} \Big[ L_k^{(\nu+1)}(x)^2-L_{k-1}^{(\nu+1)}(x)^2 \Big] =  -e^{-x} L_{N-1}^{(\nu+1)}(x)^2. 
\end{align*} 
This completes the proof. 
\end{proof}

\begin{prop} \label{Prop_Psi alternative}
Let $\Psi_N$ be defined as in \eqref{def of PsiN}. Then it admits the
following alternative representation:
\begin{equation} \label{def of PsiN v2}
 \Psi_N(x) = e^{-x}\sum_{k=0}^{N-1}(N-k)  L_{k}^{(-1)}(x)^2. 
\end{equation}
Moreover, $\Psi_N$ can also be written in the integral form
\begin{equation}  \label{def of PsiN v3}
\Psi_N(x) =  N\int_x^\infty e^{-u}L_{N-1}^{(0)}(u)^2 \,du -N(N-1)e^{-x}\Bigl( L_{N-1}^{(0)}(x)^2-  L_{N-2}^{(0)}(x)\,L_{N}^{(0)}(x) 
\Bigr). 
\end{equation}
\end{prop}
 
\begin{proof} We first establish \eqref{def of PsiN v2}. 
By Proposition~\ref{Prop_BH relation} together with \eqref{def of rho}, 
it suffices to verify the differential identity
\begin{align} \label{diff identity Laguerre in prop}
\frac{d}{dx} \Big[ e^{-x}\sum_{k=0}^{N-1}(N-k)  L_{k}^{(-1)}(x)^2 \Big ]= - e^{-x} \sum_{k=0}^{N-1} L_k^{(0)}(x)^2. 
\end{align}
For this, we rewrite the weighted sum by exchanging the order of summation,  which allows us to express it as a double sum over a triangular index set:  
\begin{align*}
\sum_{k=0}^{N-1}(N-k)  L_{k}^{(-1)}(x)^2 &=  \sum_{k=0}^{N-1} \sum_{m=k+1}^N L_{k}^{(-1)}(x)^2 = \sum_{m=1}^N \sum_{k=0}^{m-1} L_{k}^{(-1)}(x)^2. 
\end{align*}
Combining this observation with Lemma~\ref{Lem_Laguerre sum diff identity},  it follows that 
\begin{align*}
 \frac{d}{dx} \Big[ e^{-x}\sum_{k=0}^{N-1}(N-k)  L_{k}^{(-1)}(x)^2 \Big ] &=    \sum_{m=1}^N  \frac{d}{dx}  \Big[  \sum_{k=0}^{m-1}  e^{-x} L_{k}^{(-1)} (x)^2 \Big ] = - \sum_{m=1}^N e^{-x} L_{m-1}^{ (0) }(x)^2, 
\end{align*}
which leads to \eqref{diff identity Laguerre in prop}. 

We next prove \eqref{def of PsiN v3}. 
Starting from the representation \eqref{def of PsiN v2}, we decompose $\Psi_N$ into two contributions: 
\begin{align} \label{PsiN in the proof}
\Psi_N(x) = N e^{-x}\sum_{k=0}^{N-1} L_{k}^{(-1)}(x)^2- e^{-x}\sum_{k=0}^{N-1} k \, L_{k}^{(-1)}(x)^2.  
\end{align}
Notice here that by using the orthogonality relation \eqref{def of orthogonality of Laguerre} together with the Christoffel–Darboux formula, we have   
\begin{equation}
\sum_{k=0}^{N-1} k \, L_k^{(-1)}(x)^2 = N(N-1)\Bigl( L_{N-1}^{(0)}(x)^2-  L_{N-2}^{(0)}(x)\,L_{N}^{(0)}(x) \Bigr).
\end{equation}
On the other hand, by Lemma~\ref{Lem_Laguerre sum diff identity}, the first contribution admits an integral representation, which coincides with the remaining term in \eqref{def of PsiN v3}. 
Combining these two evaluations completes the proof.
\end{proof}

\section{Expansion coefficients of the Laguerre polynomials} \label{Appendix_Laguerre coeff}

In this appendix, we recall the recursive formulas developed in~\cite{FW88}, which allow one to compute the expansion coefficients appearing in Lemma~\ref{Lem_Laguerre asymp}.

Let \(\nu=4n+2\alpha+2\). Recall that \(\xi\) is defined in~\eqref{def xi}, and that \(\zeta\) is defined in~\eqref{def of zeta(x) for Airy}. It follows from~\cite[Eqs. (4.7), (5.13)]{FW88} that the Bessel-type expansion~\eqref{asymp bessel} extends to arbitrary order \(p\ge0\) in the form
\begin{equation} \label{FW laguerre bessel}
    2^\alpha e^{\nu t/2}L_n^{(\alpha)}(\nu t) = \frac{J_\alpha(\nu\xi(t))}{\xi(t)^\alpha}\sum_{k=0}^{\lfloor (p-1)/2 \rfloor}\beta_{2k}\Big(\frac{2}{\nu}\Big)^{2k} - \frac{J_{\alpha+1}(\nu\xi(t))}{\xi(t)^{\alpha+1}}\sum_{k=0}^{\lfloor p/2-1 \rfloor}\beta_{2k+1}\Big(\frac{2}{\nu}\Big)^{2k+1} + \frac{\textup{env}J_\alpha(\nu t)}{\xi(t)^{\alpha}} O(\nu^{-p}),
\end{equation}
while the Airy-type expansion~\eqref{asymp airy} extends as
\begin{equation} \label{FW laguerre airy}
\begin{split}
    (-1)^n2^\alpha e^{\nu t/2}L_n^{(\alpha)}(\nu t) &= \Ai(\nu^{2/3}\zeta(t))\sum_{k=0}^{\lfloor (p-1)/2 \rfloor}\frac{\gamma_{2k}}{\nu^{2k+1/3}} - \Ai'(\nu^{2/3}\zeta(t))\sum_{k=0}^{\lfloor p/2-1 \rfloor}\frac{\gamma_{2k+1}}{\nu^{2k+5/3}} \\
    &\quad+ \begin{cases} \displaystyle
        \textup{envAi}_{-1/4}(\nu^{2/3}\zeta(t)) O(\nu^{-p-1/3}) &\textup{if }p\textup{ is even}, \smallskip \\
        \textup{envAi}_{1/4}(\nu^{2/3}\zeta(t)) O(\nu^{-p-2/3}) &\textup{if }p\textup{ is odd}.
    \end{cases} 
\end{split}
\end{equation}
Here, $\textup{env}J_\alpha$ and $\textup{envAi}_{p}$ are given by \eqref{def of J envelop} and \eqref{def of envAi}, respectively. 

The coefficients \(\beta_k\) and \(\gamma_k\) can then be computed systematically as follows.  

Fix \(t\in(0,1)\). We first describe the coefficients \(\beta_k\) in \eqref{FW laguerre bessel}. Define
\begin{equation}
    h(u) := \Big(\frac{\sinh z(u)}{u}\Big)^{-\alpha-1}\sqrt{\frac{2u}{z(u)}}\frac{dz}{du},
\end{equation}
where \(z=z(u)\) is determined implicitly by the transformation
\begin{equation}
    z-t\coth z = u-\frac{\xi(t)^2}{u}.
\end{equation}
Next, define recursively
\begin{equation}
    h_0(u) := h(u), \qquad
    h_{n+1}(u) := g_n'(u)-\frac{\alpha+1}{u}g_n(u),
\end{equation}
where
\begin{equation}
\begin{split}
    g_{2n}(u) := u^2\frac{h_{2n}(u)-h_{2n}(i\xi(t))}{u^2-(i\xi(t))^2}, \qquad 
    g_{2n+1}(u) := u\frac{uh_{2n+1}(u)-i\xi(t)h_{2n+1}(i\xi(t))}{u^2-(i\xi(t))^2}.
\end{split}
\end{equation}
Then the coefficients \(\beta_k\) are given by
\begin{equation}
    \beta_{2n} = h_{2n}(i\xi(t)), \qquad \beta_{2n+1} = i\xi(t)h_{2n+1}(i\xi(t)).
\end{equation}
Comparing~\eqref{FW laguerre bessel} with~\eqref{asymp bessel}, we obtain
\begin{equation}
    E_1=-\frac{2\beta_1}{\xi\beta_0}, \qquad E_2=\frac{4\beta_2}{\beta_0}.
\end{equation}
After a straightforward computation, this yields the expressions stated in Lemma~\ref{Lem_Laguerre asymp} (i).

Next, we describe the coefficients \(\gamma_k\) in~\eqref{FW laguerre airy}. Define
\begin{equation}
    \mathsf{h}(u) := \Big(1- \mathsf{z} (u)^2\Big)^{(\alpha-1)/2}\frac{d \mathsf{z} }{du},
\end{equation}
where \(\mathsf{z}=\mathsf{z}(u)\) is determined by
\begin{equation}
    \frac{1}{4}\log\Big(\frac{1+\mathsf{z} }{1-\mathsf{z} }\Big)-\frac{1}{2}\mathsf{z} t = \frac{u^3}{3}-\zeta(t)u.
\end{equation}
We then define recursively
\begin{equation}
     \mathsf{h}_0(u) :=  \mathsf{h}(u), \qquad
     \mathsf{h}_{n+1}(u) := \mathsf{g}_n'(u),
\end{equation}
where
\begin{equation}
\begin{split}
    \mathsf{g}_{2n}(u) := \frac{ \mathsf{h}_{2n}(u)- \mathsf{h}_{2n}(\sqrt{\zeta(t)})}{u^2-\zeta(t)}, \qquad
    \mathsf{g}_{2n+1}(u) := \frac{ u }{u^2-\zeta(t)} \Big( \frac{\mathsf{h}_{2n+1}(u)}{u}- \frac{ \mathsf{h}_{2n+1}(\sqrt{\zeta(t)}) }{ \sqrt{\zeta(t)} } \Big) .
\end{split}
\end{equation}
The coefficients \(\gamma_k\) are then given by
\begin{equation}
    \gamma_{2n} = \mathsf{h}_{2n}(\zeta(t)^{1/2}), \qquad \gamma_{2n+1} = \mathsf{h}_{2n+1}(\zeta(t)^{1/2})/\zeta(t)^{1/2}.
\end{equation}
Comparing \eqref{FW laguerre airy} with \eqref{asymp airy}, one sees that
\begin{equation}
    F_1=-\frac{\gamma_1}{\gamma_0}, \qquad F_2=\frac{\gamma_2}{\gamma_0}.
\end{equation}
A straightforward computation then gives the expressions in Lemma~\ref{Lem_Laguerre asymp} (ii).

\bibliographystyle{abbrv}

\end{document}